\documentclass{IEEEtran}
\usepackage[latin9]{inputenc}
\usepackage{float}
\usepackage{color}
\usepackage{amsmath}
\usepackage{amsthm}
\usepackage{amssymb}
\usepackage{graphicx}

\makeatletter

\floatstyle{ruled}
\newfloat{algorithm}{tbp}{loa}
\providecommand{\algorithmname}{Algorithm}
\floatname{algorithm}{\protect\algorithmname}

  \theoremstyle{definition}
  \newtheorem{defn}{\protect\definitionname}
  \theoremstyle{plain}
  \newtheorem{thm}{\protect\theoremname}
  \theoremstyle{plain}
  \newtheorem{corollary}{Corollary}
  \theoremstyle{plain}
  \newtheorem{lem}{\protect\lemmaname}
  \theoremstyle{definition}
  
  \theoremstyle{remark}
  \newtheorem{rem}{\protect\remarkname}
  \theoremstyle{plain}
  \newtheorem{prop}{\protect\propositionname}
 
  \theoremstyle{assumption}

\@ifundefined{date}{}{\date{}}
\usepackage{algorithm}
\usepackage{algorithmic}

\providecommand{\definitionname}{Definition}
\providecommand{\lemmaname}{Lemma}
\providecommand{\problemname}{Problem}
\providecommand{\propositionname}{Proposition}
\providecommand{\remarkname}{Remark}
\providecommand{\theoremname}{Theorem}

\providecommand{\definitionname}{Definition}
\providecommand{\lemmaname}{Lemma}
\providecommand{\problemname}{Problem}
\providecommand{\propositionname}{Proposition}
\providecommand{\remarkname}{Remark}
\providecommand{\theoremname}{Theorem}

\makeatother

\providecommand{\definitionname}{Definition}
\providecommand{\lemmaname}{Lemma}
\providecommand{\problemname}{Problem}
\providecommand{\propositionname}{Proposition}
\providecommand{\remarkname}{Remark}
\providecommand{\theoremname}{Theorem}

\begin{document}

\title{The Vulnerability of Cyber-Physical System under Stealthy Attacks}
\author{Tianju Sui$^{1}$, Yilin Mo$^{2,\dagger}$ Damián Marelli$^{3}$, Ximing Sun$^{1}$ and Minyue Fu$^{4}$,
\emph{Fellow IEEE}\thanks{$^{1}$Tianju Sui and Ximing Sun are with the the School of Control Science and Engineering,
Dalian University of Technology, Dalian, China. Email:{\footnotesize{}
}\texttt{\footnotesize{}suitj@dlut.edu.cn;~sunxm@dlut.edu.cn.}}\thanks{$^{2}$Yilin Mo is with the Department of Automation, Tsinghua University, Beijing, China. Email:{\footnotesize{}
}\texttt{\footnotesize{}ylmo@tsinghua.edu.cn.}}\thanks{$^{3}$Damián Marelli is with the School of Automation, Guangdong University
of Technology, Guangzhou, China, and with the French Argentine International
Center for Information and Systems Sciences, National Scientific and
Technical Research Council, Argentina. Email: \texttt{\footnotesize{}Damian.Marelli@newcastle.edu.au}{\footnotesize{}.}}\thanks{$^{4}$Minyue Fu is with the School of Electrical Engineering and
Computer Science, University of Newcastle, Callaghan, NSW 2308, Australia.
He also holds an Qianren Professorship at the School of Automation,
Guandong University of Technology, China. Email:{\footnotesize{} }\texttt{\footnotesize{}Minyue.Fu@newcastle.edu.au.}}\thanks{$\dagger$Corresponding author.}}
\maketitle
\begin{abstract}
In this paper, we study the impact of stealthy attacks on the Cyber-Physical System (CPS) modeled as a stochastic linear system. An attack is characterised by a malicious injection into the system through input, output or both, and it is called stealthy (resp.~strictly stealthy) if it produces bounded changes (resp.~no changes) in the detection residue. Correspondingly, a CPS is called vulnerable (resp.~strictly vulnerable) if it can be destabilized by a stealthy attack (resp.~strictly stealthy attack). We provide necessary and sufficient conditions for the vulnerability and strictly vulnerability. For the invulnerable case, we also provide a performance bound for the difference between healthy and attacked system. Numerical examples are provided to illustrate the theoretical results.
%
\end{abstract}


\section{Introduction\label{sec:intro}}
Cyber-Physical Systems (CPSs), such as sensor networks, smart grids and transportation systems, are widely used in applications. Such a system combines a physical system with network technology to greatly improve the efficiency of the system. However, at the same time, this combination increases the vulnerability of the system. In particular, CPSs are subject to possible cyber attacks. At the physical system level, such attacks are characterised by malicious injections into the system through input, output or both.

The Stunex attack is one of the most famous CPS attack till now~\cite{Farwell2011}. In June 2010, a targeted virus was injected into the Bushehr nuclear power plant through a USB flash disk. The virus replaced the measurement data from the centrifuges by a sequence of ``normal" data to mislead the fault detection system to trust that the system was operating normally. Then, the virus injected input signals to accelerate the centrifuges to self destruction. This incident was reported to have caused a series of disastrous effects and destroyed over 3000 centrifuges~\cite{Farwell2011}. The attacks like Stunex may penetrate the traditional information protection framework(such as FDI) of CPS.

Other examples of CPS attacks include: the Maroochy water breach~\cite{Slay2007}, the blackout in brazil power grid~\cite{Conti2010}, the SQL Slammer attack in Davis-Besse nuclear power plant~\cite{Kuvshinkova2003}, and many other industry security incidents~\cite{Richards2008}. According to the statistical data from ICS-CERT (see https://ics-cert.us-cert.gov), there were 245 CPS attacks confirmed in 2014 and the number increased to 295 in 2015.

CPS security has attracted many researchers to focus on this area~\cite{Sastry2008chanllenge}. The traditional efforts, such as robust statistics~\cite{Huber2011} and robust control~\cite{Zhou1996}, are designed to withstand certain types of failures. The popular Fault Detection and Isolation (FDI) method assumes that the failure is spontaneous~\cite{Willsky1976,Massou1989FDI,Hwang2010FDI}. However, CPS attacks are usually purposely designed to be stealthy and destructive, and are often done with the full or partial knowledge of the system dynamic model. Thus, it is insufficient to rely on robust control or FDI against CPS attacks. As shown in \cite{Liu2011attack}, an attacker can take advantage of the configuration of a power system to launch such attacks to successfully bypass the existing techniques for bad measurement detection.

For a CPS with a linear dynamic model for the physical system, many studies have been done in the detection and analysis of malicious attacks. The work of~\cite{Pasqualetti2010} studied the performance of an average consensus algorithm when individual agents in a networked system are under attack. In~\cite{Pasqualetti2013}, the authors studied the detectability of attacks and pointed out that, for a noiseless system, the only undetectable space for attacks is due to the unknown initial state. In \cite{Sundaram2010}, an algorithm was offered to detect attacked sensors in a multi-sensors network. The work of \cite{Fawzi2012} analyzed the performance of an attacked system and studied the stabilization problem using state feedback.

The above works all assumed that the physical system is noiseless, which is very restrictive. System noises would give a shelter for attacks because they may be mistaken for noises. For static systems subject to noises, \cite{Wagner2004} utilized a general evaluation standard to study the robustness of the network cluster mechanism against attacks;  In ~\cite{Liu2011false}, the authors considered the estimation problem in a smart grid and studied how does an undetectable attack change the state of the system.

For dynamic systems subject to noises, \cite{Mo2010false} studied the performance of Kalman Filter under attacks. They further studied the attack strategy and calculated the miss/false alarm rates of a $\chi^2$ attack detector \cite{mo2016performance}. In \cite{mo2015secure}, they also worked on the design of robust estimators against attacks for multi-sensors systems. Besides, \cite{Jin2017SenAcu} develop an adaptive controller that guarantees uniform ultimate boundedness of the closed-loop dynamical system in the face of adversarial sensor and actuator attacks. The works in \cite{Jin2018Distur,Jin2018multi} extend the results to the cyber-physical systems subject to exogenous disturbances and leader-follower multiagent systems, respectively.

A lot of studies have also been done on special types of attacks. Zhang {\em et.~al.} focused on the energy-constrained attack scheduling problem for Denial-of-Service (DoS) attacks~\cite{ZhangH2015}. Zhao {\em et.~al.} studied the effect of stealthy attacks on consensus-based distributed economic dispatch~\cite{ZhaoC2017}. Kung {\em et.~al.} defined an $\epsilon$-stealthy attack  and analyzed its effect for scalar systems~\cite{Kung2017}. In \cite{DingK2017}, the authors worked on the multi-channel transmission schedule problem for remote state estimation under DoS attacks.

In this paper, we focus on a stochastic linear system under both sensor and actuator attacks. Firstly, we consider stealthy attacks or strictly stealthy attacks whose corresponding effect on the detection residue is either bounded or zero. It is noted that a stealthy attack is practically difficult to detect and a strictly stealthy attack is theoretically impossible to detect. We then study system's vulnerability under such attacks.  A system is said to be vulnerable if it can be destabilized by a stealthy attack, or strictly vulnerable if it can be destabilized by a strictly stealthy attack. We give the necessary and sufficient conditions for both vulnerable systems and strictly vulnerable systems. To further study the performance of invulnerable system under stealthy attacks, we give a performance bound for the difference between healthy system and attacked system. These results will help to understand what kind of systems are robust to stealthy attacks and how to reduce their impact on the performance.

Focusing on a standard stochastic linear system equipped with state feedback controller and Romberg state observer, the contributions of this paper are mainly in two-folds: 1) The necessary and sufficient strict vulnerability/vulnerability conditions are given. The designers of Cyber-physical systems can check the robustness of system under stealthy attacks and understand what sensor/actuator channels are critical to the vulnerability; 2) A universal upper bound for the performance is given. The designers of Cyber-physical systems can evaluate the damage caused by attacks.

The rest of this paper is organized as follows. In Section~\ref{sec:problem}, we describe the models of CPS and attacks under our study. In Section~\ref{subsec:detect}, we introduce the definitions of stealthy and strictly stealthy attacks. The notions of vulnerability and strictly vulnerability are defined according to the destabilizability of stealthy and strictly stealthy attacks. The necessary and sufficient conditions for strict vulnerability and vulnerability are given in Sections~\ref{subsec:sv} and \ref{subsec:v}, respectively. The invulnerable system's performance bound for the difference between healthy and attacked systems is given in Section~\ref{sec:bound}. In Section~\ref{sec:simulate}, examples are given to illustrate the theoretical results. Concluding remarks are stated in Section~\ref{sec:conclude}. Some proofs are left in the Appendix.

\section{Problem Formulation\label{sec:problem}}

\subsection{System Model}
In this paper, the Cyber-physical system is modeled as a linear discrete-time stochastic system in state-space form
\begin{align}
x_{t+1} & =Ax_{t}+Bu_{t}+w_{t},\label{eq:ss1}\\
y_{t} & =Cx_{t}+v_{t},\label{eq:ss2}
\end{align}
where the state $x_t\in \mathbb{R}^n$, the measurement $y_t\in\mathbb R^m$ and the control input $u_t\in\mathbb R^p$. The process noise $w_{t} \in \mathbb{R}^n$ and the measurement noise $v_t\in \mathbb{R}^m$ obey some zero-mean stochastic distributions. Moreover, $A\in \mathbb{R}^{n\times n}$ is the system matrix, $B\in \mathbb{R}^{n\times p}$ is the actuator matrix and $C\in \mathbb{R}^{m\times n}$ is the measurement matrix. In the rest of paper, it is assumed that $(A,C)$ is observable and $(A,B)$ is controllable.

Furthermore, the control input $u_t$ is assumed to be generated by a steady-state controller. To be specific, a steady-state controller is given by
\begin{eqnarray}
u_t = L\hat{x}_t,\label{lqg}
\end{eqnarray}
where $\hat{x}_t$ is generated by the estimator in \eqref{kf} below and $L\in \mathbb{R}^{p\times n}$ is chosen such that $A+BL$ is stable.

We assume a linear time-invariant Luenberger estimator is being deployed, which has the following form:
\begin{eqnarray}
\hat{x}_{t+1} = A\hat{x}_t+Bu_t+K[y_{t+1}-C(A\hat{x}_t+Bu_t)],\label{kf}
\end{eqnarray}
where $K$ is chosen such that $A-KCA$ is stable.

\begin{rem}
Other than the constraint that both $A+BL$ and $A-KCA$ are stable, the choices of $L$ and $K$ are arbitrary.
\end{rem}

We define the innovation signal $z_t$ as
\begin{eqnarray}
z_{t+1} = y_{t+1}-C(A\hat{x}_t+Bu_t), \label{eq:z_k}
\end{eqnarray}
and the estimation error $e_t$ as
$$e_{t} = x_t - \hat{x}_t.$$

Combining \eqref{eq:ss1} and \eqref{kf}, one can prove that $e_t$ follows the following recursive equation:
\begin{eqnarray}
e_{t+1} = (A-KCA)e_t+(I-KC)w_t-Kv_{t+1}.
\end{eqnarray}

\subsection{Attack Model}

In this paper, we assume that the adversary can inject an external control input and manipulate a subset of the sensory data. Therefore, system under the attack can be described by the following equations:
\begin{eqnarray}
x_{t+1}' &=& Ax_{t}'+Bu_{t}'+B^au^a_t+w_{t},\label{eq:at1}\\
y_{t}' &=& Cx_{t}'+\Gamma^ay^a_t+v_{t},\label{eq:at2}
\end{eqnarray}
where we use $(\cdot)'$ to denote the variable $\cdot$ under attack, $u^a_t\in \mathbb{R}^{p_a}$ is the actuator attack signal, $y^a_t\in \mathbb{R}^{m_a}$ is the sensor attack signal\footnote{In this paper, we do not put any constraint on $u^a_t$ and $y^a_t$ except that they need to satisfy stealthy or strictly stealthy requirement, which is introduced later in Section~\ref{subsec:detect}.}, $B^a\in \mathbb{R}^{n\times p_a}$ is the actuator attack matrix and $\Gamma^a\in \mathbb{R}^{m\times m_a} = \begin{bmatrix}
e_{i_1}&\ldots&e_{i_{m_a}}
\end{bmatrix}$ is the sensor attack matrix, where $e_i$ are the $i$th canonical basis vector of $\mathbb R^m$, and $\{i_1,\ldots,i_{m_a}\}$ is the set of the compromised sensors. Moreover, the attack is assumed to start at time $1$.

Without loss of generality, we assume that both $B^a$ and $\Gamma^a$ are full column rank.\footnote{If $B^a$ is not full column rank, then certain column of $B^a$ can be represented by a linear combination of other columns, i.e., the effect of certain malicious actuator on the system can be duplicated by the combined effect of several other malicious actuators. Therefore, removing the redundant actuator and corresponding column in $B^a$ will not change the attacker's capability.} The dimension $p_a$ and $m_a$ of the attack signal $u^a_t$, $y^a_t$ represent the attacks' degrees of freedom.

In order to consider the worst-case scenario for the CPS, the attacker is assumed to know the full system model (\ref{eq:ss1})-(\ref{eq:ss2}).

In the presence of the adversary, the steady-state estimator and controller are given by
\begin{eqnarray}
\hat{x}_{t+1}' &=& A\hat{x}_{t}'+Bu_{t}'+K[y_{t+1}' - C(A\hat{x}_t'+Bu_t')],\label{eq:akf1}\\
u_t' &=& L\hat{x}_t'.\nonumber
\end{eqnarray}


The innovation signal and estimation error are updated as
\begin{eqnarray}
z_{t+1}' &=& y_{t+1}' - C(A\hat{x}_t'+Bu_t'),\label{eq:z_kat}\\
e_t' &=& x_t'-\hat{x}_t'.
\end{eqnarray}

The difference between an attacked system and the healthy system is characterized by
\begin{align}
  \triangle x_t &\triangleq x_t' - x_t,&
                     \triangle\hat{x}_t &\triangleq \hat{x}_t' - \hat{x}_t,\nonumber\\
  \triangle u_t &\triangleq u_t' - u_t,&
                             \triangle y_t &\triangleq y_t' - y_t,\nonumber\\
  \triangle z_t &\triangleq z_t' - z_t,&
                             \triangle e_t &\triangleq e_t' - e_t.\label{eq:tri}
\end{align}
The difference variables are of particular interest for the adversary and will be the focus in the rest of the paper. To be specific, $\Delta z_t$ and $\Delta y_t$ can be used to characterize the stealthiness of the attack. An intrusion detector employed by the CPS is unable (or hardly able) to distinguish a healthy system and a compromised system if $\triangle z_t$ and $\triangle y_t$ are zero (or small enough). The quantities $\triangle x_t$ and $\triangle e_t$ can be used to quantify the damage caused by the attack.
\begin{rem}
Since we assume that the attacks start at time 1, the biases between healthy and attacked system are all zeroes at time 0, i.e., $\triangle e_0=0, \triangle x_0=0$ and $\triangle z_0=0$.
\end{rem}

\section{Classifications for systems and attacks}\label{subsec:detect}
In this section, we shall classify the attacks depending on the stealthiness of the attack. Since the input of any detector is the sensory data $\{y_t':t\in \mathbb{N}\}$, and there is a one-to-one mapping between the residual error sequence $\{z_t':t\in \mathbb{N}\}$ and the sensory data, we analyze the difference between the attacked system's $z_t'$ and the healthy system's $z_t$, i.e., $\triangle z_t$, to determine if an attack can be detected or not. An attack is impossible to be detected if
\begin{eqnarray}\label{eq:detect1}
\|\triangle z_t\|= 0, \forall t\in\mathbb{N}.
\end{eqnarray}

In practice, an attack is hardly detectable if $\triangle z_t$ is small enough, i.e., there exist $\delta>0$ such that
\begin{eqnarray}\label{eq:detect2}
\|\triangle z_t\|\leq \delta, \forall t\in\mathbb{N}.
\end{eqnarray}
\begin{rem}
As proved by Theorem~1 in \cite{mo2016performance}, for a linear Gaussian system monitored by a $\chi^2$ detector, the alarm rate converges to the false alarm rate as $\delta\rightarrow 0$. Moreover, Bai et al.~\cite{bai2015cdc} have proven a similar result for other forms of detectors.
\end{rem}
Based on the above, the classification of attacks is given below.
\begin{defn}\label{defn:stealthy}
An attack sequence is said to be {\em stealthy} if \eqref{eq:detect2} is satisfied for some $\delta>0$ and {\em strictly stealthy} if \eqref{eq:detect1} holds.
\end{defn}
By substracting \eqref{eq:z_kat} from \eqref{eq:z_k} and \eqref{eq:akf1} from \eqref{kf}, we have
\begin{eqnarray}
\triangle \hat{x}_{t+1} &=& (A+BL)\triangle \hat{x}_{t} + K\triangle z_{t+1},\label{eq:delhatx}\\
\triangle y_{t+1} &=& \triangle z_{t+1} + C(A+BL)\triangle \hat{x}_{t}\label{eq:dely}.
\end{eqnarray}

Based on the definitions of $\triangle z_t$ and $\triangle e_t$, their update equations are given by
\begin{eqnarray}
&&\triangle e_{t+1} \nonumber\\
&=& (I-KC)A\triangle e_{t} + (I-KC)B^au_t^a - K\Gamma^ay_{t+1}^a,\ \ \ \  \label{eq:deltE}
\end{eqnarray}
and
\begin{eqnarray}
\triangle z_{t+1} &=& \triangle y_{t+1}-C(A+BL)\triangle \hat{x}_{t}\nonumber\\
&=& CA\triangle e_{t}+CB^au_t^a+\Gamma^ay_{t+1}^a.\label{eq:deltzk}
\end{eqnarray}

Noted that both $\triangle z_{t}$ and $\triangle e_{t}$ depend only on the attack signals.

A system is resilient under the attacks if both $\limsup_{t\rightarrow \infty}\|\triangle x_t\|<\infty$ and $\limsup_{t\rightarrow \infty}\|\triangle e_t\|<\infty$. The following lemma shows that we only need to check one of the conditions instead of both.
\begin{lem}\label{lem:equi}
For a stealthy or strictly stealthy attack on the system~\eqref{eq:ss1}-\eqref{eq:ss2}, a necessary and sufficient condition for $\limsup_{t\rightarrow \infty}\|\triangle x_t\|<\infty$ is
\begin{eqnarray}
\limsup_{t\rightarrow \infty}\|\triangle e_t\|<\infty.
\end{eqnarray}
\end{lem}
\begin{proof}
Based on the notations in \eqref{eq:tri}, we have
\begin{eqnarray}
\triangle x_t &=& (\hat{x}_t'+e_t')-(\hat{x}_t+e_t)= \triangle \hat{x}_t + \triangle e_t.\label{eq:equivalent}
\end{eqnarray}

Recall the equation \eqref{eq:delhatx}, we have
\begin{eqnarray*}
\triangle \hat{x}_{t+1} = (A+BL)\triangle \hat{x}_t + K\triangle z_{t+1}.
\end{eqnarray*}

Since $A+BL$ is stable and $\Delta z_k$ is bounded due to the stealthy (or strictly stealthy) requirement in \eqref{eq:detect1} and \eqref{eq:detect2}, the variable $\triangle \hat{x}_{t}$ is bounded for any $t\in \mathbb{N}$.
Therefore, the condition $\limsup_{t\rightarrow \infty}\|\triangle e_t\|<\infty$ is equivalent to  $\limsup_{t\rightarrow \infty}\|\triangle x_t\|<\infty$.
\end{proof}
In the rest of paper, we will use the boundness of $\limsup_{t\rightarrow \infty}\|\triangle e_t\|$ to represent the resilience under attacks.
Combining with the classification of attacks in Definition~\ref{defn:stealthy}, we can classify a system \eqref{eq:ss1}-\eqref{eq:ss2} depending on if there exists a stealthy (or strictly stealthy) attack to introduce an unbounded estimation error $\Delta e_t$ (or bias on the state $\Delta x_t$).
\begin{defn}\label{defn:vulner}
The linear system in \eqref{eq:ss1}-\eqref{eq:ss2} is said to be {\em vulnerable} (or {\em strictly vulnerable}) if, for any $M_1>0$, there exists a {\em stealthy} (or {\em strictly stealthy}) attack such that
\begin{eqnarray}
\limsup_{t\rightarrow \infty}\|\triangle e_t\|> M_1.\label{eq:vulner1}
\end{eqnarray}

And the system is {\em invulnerable} (or {\em strictly invulnerable}) if there exists $M_2>0$ such that
\begin{eqnarray}
\limsup_{t\rightarrow \infty}\|\triangle e_t\|\leq M_2\label{eq:invulner1}
\end{eqnarray}
for any {\em stealthy} (or {\em strictly stealthy}) attacks.
\end{defn}
\begin{rem}
The vulnerability and strict vulnerability of a system are important concepts for system security. That a system is strictly invulnerable means that it is always stable under any attacks that have no influence on the residue. Meanwhile, that a system is invulnerable means that it is always stable under any attacks that have bounded influence on the residue. Thus, the invulnerable system is more robust to the attacks. We will provide necessary and sufficient conditions for vulnerability and strict vulnerability in the following Sections~\ref{subsec:sv} and~\ref{subsec:v}, respectively.
\end{rem}
\section{The Necessary and Sufficient Condition for Strict Vulnerability}\label{subsec:sv}
This section is devoted to the characterization of strictly vulnerable systems.
\begin{defn}\label{def:invertible}
Consider the following linear system with initial state $x_0=0$:
\begin{eqnarray}
x_{k+1}=Ax_k+Bu_k,y_k=Cx_k+Du_k.\label{eq:standard}
\end{eqnarray}

The above system is said to be {\em invertible} if $y_k=0$ for all $k\in \mathbb{N}$ implies that $u_k=0$ for all $k\in \mathbb{N}$\footnote{The invertible system defined here is called {\em left invertible} in \cite{wonham1985,marro2010}.}.
\end{defn}
\begin{rem}
If a system is invertible, the mapping from the input $\{u_k:k\in \mathbb{N}\}$ to the output $\{y_k:k\in \mathbb{N}\}$ is injective, which means that different input will result in different output. In particular, any non-zero input will result in an non-zero output. 
\end{rem}

In particular, one can check the invertibility of a linear system using the following rank conditions, the proof of which can be found in \cite{Sain1969}.
\begin{prop}\label{prop:invertible}
The linear system in \eqref{eq:standard} is {\em invertible} if and only if
\begin{eqnarray}
\text{rank}(M_n)-\text{rank}(M_{n-1})=\dim(u_k),\label{eq:rankM}
\end{eqnarray}
where
\begin{eqnarray*}
M_i = \begin{bmatrix}D&0&0&\cdots&0\\ CB&D&0&\cdots&0\\ CAB&CB&D&\cdots&0\\ \vdots&\vdots&\vdots&\ddots&\vdots\\ CA^{i-1}B&CA^{i-2}B&CA^{i-3}B&\cdots&D\end{bmatrix}
\end{eqnarray*}
and $n$ is the dimension of state $x_k$.
\end{prop}
Furthermore, a complementary lemma is given to show the invertibility equivalence of two systems.
\begin{lem}\label{lem:eq}
The system in \eqref{eq:standard} is not invertible if and only if
\begin{eqnarray}
x_{k+1}'=(A+KC)x_k'+(B+KD)u_k',y'_k=Cx_k'+Du_k'.\label{eq:standardK}
\end{eqnarray}
is not invertible for any gain matrix $K$.

Furthermore, if $\{u_k\}$ is a non-zero sequence of input for system \eqref{eq:standard} such that $y_k = 0$ for all $k$, then $u_k' = u_k$ is a sequence of non-zero input for system \eqref{eq:standardK} such that $y_k' = 0$ for all $k$.
\end{lem}
\begin{proof}
See Appendix~\ref{app:eq}.
\end{proof}

Before giving the necessary and sufficient condition for strict vulnerability, we need one additional lemma:
\begin{lem}
\label{lem:noninvertibledestable}
Suppose the system \eqref{eq:standard} is non-invertible, and $\begin{bmatrix}
    B\\
    D
  \end{bmatrix}$ has full column rank, then there exists a non-zero input sequence $\{u_k\}$, such that the following holds:
  \begin{align}
\limsup_k \|x_k\|\rightarrow \infty,\text{ and } y_k = 0,\forall k.
  \end{align}
\end{lem}
\begin{proof}
Since the system is non-invertible, there exists a non-zero sequence of input $\{u_k\}$, such that the corresponding $y_k = 0$ for all $k$. Without loss of generality, we shall assume that $u_0 \neq 0$, otherwise we can always trim the leading zero inputs in the sequence $\{u_k\}$. Now notice that
\begin{align}
\begin{bmatrix}
x_1\\
y_0
\end{bmatrix}  = \begin{bmatrix}
  A\\
  C
\end{bmatrix} x_0 + \begin{bmatrix}
  B\\
  D
\end{bmatrix} u_0 =  \begin{bmatrix}
  B\\
  D
\end{bmatrix} u_0 .
  \end{align}
  The fact that $ \begin{bmatrix}
  B\\
  D
\end{bmatrix}$ has full column rank and $u_0\neq 0$ implies that one of $x_1$ and $y_0$ is non-zero. Since $y_k$ is constantly $0$, we conclude that $x_1\neq 0$. Without loss of generality, by proper scaling of $u_k$, we can assume that $\|x_1 \| = 1$.

Now if $\limsup_k \|x_k\| \rightarrow \infty$, then we finish the proof. Otherwise, suppose that $\sup_k \|x_k\| \leq M$. We can recreate an input sequence \footnote{The design of $u'_k$ utilizes the linearity combination property~\cite{chen1998linear}, which guarantees both $y'_k=0$ for all $k\in \mathbb{N}$ and $x'_k=\sum_{i=0}^k(2M+1)^{k-i}x_k$. Moreover, the item $(2M+1)$ is used to ensure the divergence of $x'_k$.} $u'_k$ , such that
\begin{eqnarray}
u'_k = \sum_{i=0}^k(2M+1)^{k-i}{u}_i.\label{eq:striv}
\end{eqnarray}

Based on the property of linear systems, the corresponding state $x'_k$ and measurement $y'_k$ satisfy that
\begin{eqnarray*}
x'_k &=& \sum_{i=0}^k(2M+1)^{k-i}x_i,\\
y'_k &=& \sum_{i=0}^k(2M+1)^{k-i}y_i = 0.
\end{eqnarray*}

Thus, we have
\begin{eqnarray*}
\|x'_k\| &\geq& (2M+1)^{k-1}\|x_1\| - \sum_{i=2}^k(2M+1)^{k-i}\|x_i\|\\
&\geq& (2M+1)^{k-1} - \sum_{i=2}^k(2M+1)^{k-i}M\\
&=& \frac{(2M+1)^{k-1}+1}{2}.
\end{eqnarray*}
It is obvious that $\limsup_{k\rightarrow \infty}\|x'_k\|=\infty$.
\end{proof}

We can now provide the necessary and sufficient condition for strict vulnerability:
\begin{thm}\label{thm:stable1}
The system \eqref{eq:ss1}-\eqref{eq:ss2} is strictly vulnerable if and only if the following system is not invertible:
\begin{eqnarray}
x_{k+1} = Ax_k + \begin{bmatrix}B^a&0\end{bmatrix}\zeta_k,\label{eq:cominv1}\\
y_{k} = CAx_k + \begin{bmatrix}CB^a&\Gamma^a\end{bmatrix}\zeta_k,\label{eq:cominv2}
\end{eqnarray}
where $\zeta_k=\begin{bmatrix}u_k^a\\ y_{k+1}^a\end{bmatrix}$ is the input of system \eqref{eq:cominv1}-\eqref{eq:cominv2}.
\end{thm}
\begin{proof}
{\bf Sufficiency:} Firstly, based on Lemma~\ref{lem:eq}, the fact that system \eqref{eq:cominv1}-\eqref{eq:cominv2} is not invertible implies that the following system is not invertible for any $K$:
\begin{eqnarray}
x_{k+1} &=& (A-KCA)x_k \nonumber\\
&&+ \begin{bmatrix}B^a-KCB^a&-K\Gamma^a\end{bmatrix}\zeta_k,\label{eq:cominv3}\\
y_k &=& CAx_k + \begin{bmatrix}CB^a&\Gamma^a\end{bmatrix}\zeta_k.\label{eq:cominv4}
\end{eqnarray}

Notice that
\begin{align*}
\begin{bmatrix}
  I &\\
  -C & I
\end{bmatrix}
  \begin{bmatrix}
  I &K\\
   & I
\end{bmatrix}\begin{bmatrix}
  B^a-KCB^a &-K\Gamma^a\\
  CB^a &\Gamma^a
\end{bmatrix} = \begin{bmatrix}
  B^a&\\
  &\Gamma^a
\end{bmatrix}
\end{align*}
Therefore, $\begin{bmatrix}
  B^a-KCB^a &-K\Gamma^a\\
  CB^a &\Gamma^a
\end{bmatrix}$ has full column rank and by Lemma~\ref{lem:noninvertibledestable}, there exists a sequence of $\zeta_k$ to make $x_k$ unbounded and $y_k = 0$ for all $k$.

Note that \eqref{eq:cominv3}-\eqref{eq:cominv4} is an alternative expression of \eqref{eq:deltE}-\eqref{eq:deltzk} for the dynamics of $\triangle e_t$ and $\triangle z_t$. Hence, there exists a strictly stealthy attack to make $\Delta z_t = 0$ and $\Delta e_t\rightarrow \infty$.

{\bf Necessity:} Suppose  the system \eqref{eq:ss1}-\eqref{eq:ss2} is strictly vulnerable, then there exists a non-zero input $\{\zeta_k:k\in \mathbb{N}\}$ such that $y_k=0$ for all $k\in \mathbb{N}$ in \eqref{eq:cominv4}. This means that the system \eqref{eq:cominv3}-\eqref{eq:cominv4} is non-invertible. Based on  Lemma~\ref{lem:eq}, the system \eqref{eq:cominv1}-\eqref{eq:cominv2} is also not invertible.
\end{proof}

We can further simplify our invertibility condition in Theorem~\ref{thm:stable1} as follows.
\begin{corollary}
The system \eqref{eq:ss1}-\eqref{eq:ss2} is strictly vulnerable if and only if the following system is not invertible:
  \begin{align}
x'_{k+1} = A x'_k + \begin{bmatrix}
B^a&0
\end{bmatrix}\zeta'_k,\,y'_{k} = C x'_k + \begin{bmatrix}
0&\Gamma^a
\end{bmatrix}\zeta'_k.\label{eq:simplifiedsys}
  \end{align}
\end{corollary}

\begin{proof}
We only need to prove that the system \eqref{eq:cominv1}-\eqref{eq:cominv2} is not invertible if and only if \eqref{eq:simplifiedsys} is not invertible:
Suppose that the input $\zeta_k$ for system \eqref{eq:cominv1}-\eqref{eq:cominv2} is of the form $\zeta_k = \begin{bmatrix}
  u^a_k\\
  y^a_{k+1}
\end{bmatrix}$. Then let
\begin{align*}
\zeta'_k = \begin{bmatrix}
  u^a_{k}\\
  y^a_{k}
\end{bmatrix},
\end{align*}
it follows that\footnote{We assume that $y_{-1} = y_{-1}^a = 0$.}
\begin{align*}
x'_k = x_k,\, y'_{k} = y_{k-1}.
\end{align*}

Therefore, the system \eqref{eq:cominv1}-\eqref{eq:cominv2} is non-invertible if and only if \eqref{eq:simplifiedsys} is non-invertible.
\end{proof}

\begin{rem}
From the viewpoint of structured linear system, one can use Theorem~2 in \cite{Dion2003} to derive the generic rank of the transfer function of the system described in \eqref{eq:simplifiedsys} and thus check if the system is left invertible or not (i.e., if the transfer function from the input to the output is full row rank or not.)
\end{rem}

\section{The Necessary and Sufficient Condition for Vulnerability}\label{subsec:v}
Next, we study the necessary and sufficient condition for vulnerability.
\begin{defn}\label{def:invariant}
The set $V$ is {\em invariant} if, for any $v\in V$, there exists $u$ such that
\begin{eqnarray}
A\upsilon + Bu\in V,C\upsilon+Du=0.\label{eq:invariant1}
\end{eqnarray}
Let $V_m$ be the maximum invariant subspace, the existence of which is proven in \cite{Anderson1975}. The {\em maximum reachable invariant set} is given by\footnote{The invariant set $V$ here is also called {\em output-nulling controlled invariant subspace} in \cite{marro2010} and \cite{Anderson1975}. Especially, the maximum reachable invariant set $V^*$ is called {\em the maximum output-nulling controlled invariant subspace}.}
\begin{eqnarray}
V^*=\text{span}\begin{pmatrix}B&AB&\ldots&A^{n-1}B\end{pmatrix}\cap V_m. \label{eq:invariant_m}
\end{eqnarray}
\end{defn}
A property of the maximum reachable invariant set is shown below.
\begin{lem}\label{lem:invariant}
If the system \eqref{eq:standard} is invertible, for each $x \in V^*$, then there exists a unique $u$ such that
\begin{align}
Ax + Bu\in V^*,\,Cx+Du = 0.  \label{xxx}
\end{align}

Moreover, there exists a matrix $Q$ such that $u=Qx$ for every pairs $(x,u)$ satisfying \eqref{xxx}.
\end{lem}
\begin{proof}
See Appendix~\ref{app:invariant}.
%
\end{proof}
Before the proof of necessary and sufficient condition for vulnerability, a notation is defined to facilitate Lemmas~\ref{lem:suf}-\ref{lem:nes}.
\begin{defn}\label{defn:zero}
The system
\begin{eqnarray}
x_{k+1}=(A+BQ)x_k,~y_k=(C+DQ)x_k\label{eq:standardQ}
\end{eqnarray}
has {\em unstable reachable zero-dynamic}, if there exist a vector $v$ satisfying the following conditions:

1) $v$ is an unstable eigenvector of $A+BQ$ and its corresponding eigenvalue is $\lambda$ with $|\lambda|\geq 1$;

2) $(C+DQ)v=0$;

3) $v$ is reachable for $(A,B)$.
\end{defn}
\begin{rem}
The {\em unstable reachable zero-dynamic} in Definition~\ref{defn:zero} contains three parts:

1) The existence of zero-dynamic space for \eqref{eq:standardQ}, which guarantees that the output (i.e., residue) is bounded;

2) The zero-dynamic space contains an unstable eigen-space of \eqref{eq:standardQ}, which makes the state diverge;

3) The reachability of $v$, which implies that a vector satisfying 1) and 2) can be reached by certain sequence of input $u_k$.
\end{rem}
Next, we provide a lemma on the zero-dynamic property to study the sufficiency conditions for vulnerability.
\begin{lem}\label{lem:suf}
Consider the system \eqref{eq:standard}, for any $M>0$, there exists a stealthy input sequence $\{u_k:k\in \mathbb{N}\}$ such that
\begin{eqnarray*}
\limsup_{k\rightarrow \infty}\|{x}_k\|&>&M,\\
\|y_k\|&\leq&\delta,\forall k\in \mathbb{N},
\end{eqnarray*}
if there exists a matrix $Q$ such that the system \eqref{eq:standardQ} has {\em unstable reachable zero-dynamic}.
\end{lem}
\begin{proof}
See Appendix~\ref{app:suf}.
\end{proof}
The following two lemmas on the conditions for zero-dynamic are needed for studying the necessity conditions for vulnerability.

\begin{lem}\label{lem:noninto}
Suppose the system \eqref{eq:standard} is non-invertible, then there exists a matrix $Q$ such that the system \eqref{eq:standardQ} has {\em unstable reachable zero-dynamic}.
\end{lem}
\begin{proof}
See Appendix~\ref{app:nes1}.
\end{proof}

\begin{lem}\label{lem:nes}
Consider an invertible system \eqref{eq:standard} with $\ker(B)\cap \ker(D)=\emptyset$. Suppose that, for any $M>0$, there exists a stealthy input sequence $\{u_k:k\in \mathbb{N}\}$ such that
\begin{eqnarray*}
\limsup_{k\rightarrow \infty}\|{x}_k\|&>&M,\\
\|y_k\|&\leq&\delta,\forall k\in \mathbb{N},
\end{eqnarray*}
then there exists a matrix $Q$ such that the system \eqref{eq:standardQ} has {\em unstable reachable zero-dynamic}.
\end{lem}
\begin{proof}
See Appendix~\ref{app:nes2}.
\end{proof}
Based on the results in Lemmas~\ref{lem:suf}-\ref{lem:nes}, the main result on vulnerability is given below.
\begin{thm}\label{thm:stable}
The system \eqref{eq:ss1}-\eqref{eq:ss2} is vulnerable if and only if there exists a vector $v$ and matrix $Q$ satisfying the following conditions:

1) $v$ is an unstable eigenvector of $A+B^aQ$ and its corresponding eigenvalue is $\lambda$ with $|\lambda|\geq 1$;

2) $Cv\in \text{span}(\Gamma^a)$;

3) $v$ is reachable for $(A-KCA,\begin{bmatrix}B^a-KCB^a&-K\Gamma^a\end{bmatrix})$.
\end{thm}
\begin{proof}
{\bf Sufficiency:}. We need to show that conditions 1)-3) imply vulnerability.

Since $Cv\in \text{span}(\Gamma^a)$, there exist a vector $y^*$ such that $Cv=\Gamma^ay^*$ and a matrix $W$ such that $\lambda y^*=-Wv$. Taking a gain matrix $\begin{bmatrix}Q\\W\end{bmatrix}$, then
\begin{eqnarray}
[CA+\begin{bmatrix}CB^a&\Gamma^a\end{bmatrix}\begin{bmatrix}Q\\W\end{bmatrix}]v &=& C(A+B^aQ)v+\Gamma^a Wv\nonumber\\
&=& \lambda Cv - \lambda \Gamma^ay^*\nonumber\\
&=& 0\label{eq:sufy}
\end{eqnarray}
and
\begin{eqnarray}
&&[(A-KCA)+\begin{bmatrix}B^a-KCB^a&-K\Gamma^a\end{bmatrix}\begin{bmatrix}Q\\W\end{bmatrix}]v \nonumber\\
&=& (A+B^aQ)v-K[C(A+B^aQ)v+\Gamma^a Wv]\nonumber\\
&=& \lambda v.\label{eq:sufx}
\end{eqnarray}

Till now, we know that the unstable eigenvector $v$ of $[(A-KCA)+\begin{bmatrix}B^a-KCB^a&-K\Gamma^a\end{bmatrix}\begin{bmatrix}Q\\W\end{bmatrix}]$ satisfies the condition 1)-2) in Definition~\ref{defn:zero} for system \eqref{eq:cominv3}-\eqref{eq:cominv4} (the same as system \eqref{eq:deltE}-\eqref{eq:deltzk} for $\triangle e_t$ and $\triangle z_t$).

Combining with that $v$ is reachable for $(A-KCA,\begin{bmatrix}B^a-KCB^a&-K\Gamma^a\end{bmatrix})$, the condition 3) in Definition~\ref{defn:zero} is satisfied. Thus, based on the Lemma~\ref{lem:suf}, the system \eqref{eq:cominv3}-\eqref{eq:cominv4} can be destabilized by a stealthy input, i.e., the system \eqref{eq:ss1}-\eqref{eq:ss2} is vulnerable.  

{\bf Necessity:} We need to show that vulnerability implies conditions 1)-3).

Note that $B^a$ and $\Gamma^a$ are both full column rank. Recall the proof of Theorem~\ref{thm:stable1}, we have already proved that $\begin{bmatrix}B^a-KCB^a &-K\Gamma^a\\CB^a &\Gamma^a\end{bmatrix}$ has full column rank, i.e., 
$$\ker\begin{bmatrix}B^a-KCB^a&-K\Gamma^a\end{bmatrix}\cap \ker\begin{bmatrix}CB^a&\Gamma^a\end{bmatrix} = \emptyset$$
for any matrix $K$.

We will separate the rest of the proof into two cases:

1) Suppose  the system \eqref{eq:cominv3}-\eqref{eq:cominv4} is non-invertible:

Following the Lemma~\ref{lem:noninto}, there exists a vector $\upsilon$ and gain matrix $\Psi$ such that

(a) $\upsilon$ is an unstable eigenvector of the following matrix
\begin{eqnarray}
A-KCA+\begin{bmatrix}B^a-KCB^a&-K\Gamma^a\end{bmatrix}\Psi.\label{eq:e1}
\end{eqnarray}

(b) $\upsilon$ satisfies that
\begin{eqnarray}
[CA+\begin{bmatrix}CB^a&\Gamma^a\end{bmatrix}\Psi]\upsilon = 0.\label{eq:e2}
\end{eqnarray}

(c) $\upsilon$ is reachable for $(A-KCA,\begin{bmatrix}B^a-KCB^a&-K\Gamma^a\end{bmatrix})$.

2) Suppose that the system in \eqref{eq:cominv3} and \eqref{eq:cominv4} is invertible:

Based on the result in Lemma~\ref{lem:nes}, (a)-(c) also holds.

Let $\Psi = \begin{bmatrix}Q\\W\end{bmatrix}$, from \eqref{eq:e2}, we have
\begin{eqnarray}
[CA+CB^aQ+\Gamma^aW]\upsilon = 0.\label{eq:e3}
\end{eqnarray}

Combining \eqref{eq:e3} with the fact that $\upsilon$ is an unstable eigenvector of \eqref{eq:e1}, we have
\begin{eqnarray}
&&[A-KCA+B^aQ-KCB^aQ-K\Gamma^aW]\upsilon \nonumber\\
&=& (A+B^aQ)\upsilon - K(CA+CB^aQ+\Gamma^aW)\upsilon \nonumber\\
&=& (A+B^aQ)\upsilon = \lambda \upsilon.\label{eq:e4}
\end{eqnarray}

Therefore, $\upsilon$ is also an unstable eigenvector of $A+B^aQ$. Recall that \eqref{eq:e3} implies
$$
\lambda C\upsilon = -\Gamma^aW\upsilon.
$$

That is, $C\upsilon\in \text{span}(\Gamma^a)$. Hence, conditions 1)-3) all hold.
\end{proof}
\begin{rem}
It is worth to note that the value of $\delta>0$ in stealthy condition \eqref{eq:detect2} is independent of the vulnerability condition. This is due to the linearity of the system. The adversary can always scale its attack to make $\Delta z_t$ arbitrarily small, while making $\Delta e_t$ diverge.
\end{rem}
\begin{rem}
Following Definition~\ref{defn:zero} and Lemmas~\ref{lem:suf}-\ref{lem:nes}, the vulnerability condition in Theorem~\ref{thm:stable} can be divided into three parts:

1) The existence of a non-trivial {\em output-nulling invariant subspace} for system \eqref{eq:cominv3}-\eqref{eq:cominv4};

2) There exists a unstable eigenvector $v$ of $(A+B^aQ)$ belonging to the above {\em output-nulling invariant subspace};

3) This unstable eigenvector $v$ is reachable for $(A-KCA,\begin{bmatrix}B^a-KCB^a&-K\Gamma^a\end{bmatrix})$.

The existence of {\em output-nulling invariant subspace} in 1) can be checked through the structural knowledge of linear system~\cite{Dion2003}. On the other hand, conditions 2)-3) are not generic properties and cannot be evaluated using structural only information. Hence, the structural information about the system can provide a necessary condition on whether the system is vulnerable.
\end{rem}

\section{A Performance Bound for Invulnerable System}\label{sec:bound}
Following the necessary and sufficient condition for vulnerability in Theorem~\ref{thm:stable}, we understand that the bias $\triangle {e}_t$ between the healthy and an attacked systems is bounded when the condition is not satisfied. In this section, focusing on the invulnerable system, we will give a performance bound for the bias $\triangle {e}_t$ under stealthy attacks.

Recall that in \eqref{eq:deltE} and \eqref{eq:deltzk}, letting $\zeta_t=\begin{bmatrix}u^a_{t}\\ y^a_{t+1}\end{bmatrix}$, we have
\begin{eqnarray*}
\triangle z_{t+1} = CA\triangle e_{t}+\begin{bmatrix}CB^a&\Gamma^a\end{bmatrix}\zeta_t,
\end{eqnarray*}
and
\begin{eqnarray}
\triangle e_{t+1} = (I-KC)A\triangle e_{t} + \begin{bmatrix}(I-KC)B^a&-K\Gamma^a\end{bmatrix}\zeta_t. \nonumber
\end{eqnarray}

By taking the $z$-transformation on both $\triangle z_{t+1}$ and $\triangle e_{t+1}$, it deduces
\begin{eqnarray}
\triangle e(z) &=& (zI-(I-KC)A)^{-1}\nonumber\\
&&\cdot\begin{bmatrix}(I-KC)B^a&-K\Gamma^a\end{bmatrix}\zeta(z) \nonumber\\
&=& T(z)\zeta(z)\label{eq:zT}
\end{eqnarray}
and
\begin{eqnarray}
\triangle z(z) &=& [CA(zI-(I-KC)A)^{-1}\nonumber\\
&&\cdot\begin{bmatrix}(I-KC)B^a&-K\Gamma^a\end{bmatrix}+\begin{bmatrix}CB^a&\Gamma^a\end{bmatrix}]\zeta(z) \nonumber\\
&=& S(z)\zeta(z),\label{eq:zz}
\end{eqnarray}
where $T(z)=(zI-(I-KC)A)^{-1}\begin{bmatrix}(I-KC)B^a&-K\Gamma^a\end{bmatrix}$ and $S(z)=CA(zI-(I-KC)A)^{-1}\begin{bmatrix}(I-KC)B^a&-K\Gamma^a\end{bmatrix}+\begin{bmatrix}CB^a&\Gamma^a\end{bmatrix}$.

Before the main result, we firstly introduce some new definitions to facilitate the analysis.
\begin{defn}
For a vector or matrix sequence $\{\xi_{t}:t\in\mathbb{N}\}$, we
use $\xi$ to denote the whole sequence. If $\xi$ is a vector sequence,
we define
\[
\left\Vert \xi\right\Vert _{\infty,2}=\sup_{t\in\mathbb{N}}\left\Vert \xi_{t}\right\Vert _{2},
\]
and if it is a matrix sequence, we define
\[
\left\Vert \xi\right\Vert _{1,\mathrm{sp}}=\sum_{t\in\mathbb{N}}\left\Vert \xi_{t}\right\Vert _{\mathrm{sp}},
\]
where $\left\Vert \cdot\right\Vert _{\mathrm{sp}}$ denotes the spectral norm, i.e.,
\[
\left\Vert \xi_{t}\right\Vert _{\mathrm{sp}}=\sup_{\|x\|_2=1}\left\Vert \xi_{t}x\right\Vert.
\]
\end{defn}
Then, a lemma is required by the proof of main result in this section.
\begin{lem}\label{lem:boundlem}
Suppose that the system in \eqref{eq:ss1}-\eqref{eq:ss2} is invulnerable and the attack $\{y_t^a,u_t^a:t\in \mathbb{N}\}$ is stealthy, then we have
\begin{equation}
\mathrm{ker}\left(T(z)\right)\supseteq\mathrm{ker}\left(S(z)\right),\text{ for all }\left|z\right|=1\label{eq:cond}
\end{equation}
and
\begin{eqnarray}
\left\Vert R\right\Vert _{1,\mathrm{sp}}<\infty,\label{eq:R1}
\end{eqnarray}
where $\|R\|_{1,sp}=\sum_{s\in \mathbb{N}}\|R_s\|_{sp}$, $R_s=\mathcal{Z}^{-1}(R(z))$, $R(z)=T(z)S^{\dag}(z)$ and $\mathcal{Z}^{-1}(\cdot)$ is the inverse $z$-transformation.
\end{lem}
\begin{proof}
See Appendix~\ref{app:bound}.
\end{proof}
%
%
%
Based on the properties for invulnerable system in Lemma~\ref{lem:boundlem}, a bound for $\triangle {e}_{k}$ is given in Theorem~\ref{thm:universal}.
\begin{thm}\label{thm:universal}
 Suppose that the system in \eqref{eq:ss1}-\eqref{eq:ss2} is invulnerable and the attack $\{y_t^a,u_t^a:t\in \mathbb{N}\}$ is stealthy, for any $k\in \mathbb{N}$, we have
\begin{eqnarray}
\|\triangle {e}_{k}\|_2 \leq \|R\|_{1,sp}\delta,
\end{eqnarray}
where $\delta$ is the stealthy bound for the residue defined in \eqref{eq:detect2}.
\end{thm}
\begin{proof}
Based on the Lemma~\ref{lem:boundlem}, it follows from \eqref{eq:cond} that
\[
\mathrm{ker}\left(T(z)\right)\supseteq\mathrm{ker}\left(S(z)\right),\text{ for all }\left|z\right|=1.
\]

It is known that $T^{\dag}(z)T(z)$ is the projection onto $\ker(T(z))$ and $S^{\dag}(z)S(z)$ is the projection onto $\ker(S(z))$.
For that $\mathrm{ker}\left(T(z)\right)\supseteq\mathrm{ker}\left(S(z)\right)$, we can pre-multiply $T^{\dag}(z)T(z)$ by $S^{\dag}(z)S(z)$ without changing the result.

Then, from \eqref{eq:zT} and \eqref{eq:zz}, we have
\begin{eqnarray}
\triangle e(z) &=& T(z)\zeta(z) = T(z)T^{\dag}(z)T(z)\zeta(z)\nonumber\\
&=& T(z)T^{\dag}(z)T(z)S^{\dag}(z)S(z)\zeta(z) \nonumber\\
&=& T(z)S^{\dag}(z)(S(z)\zeta(z))\nonumber\\
&=& R(z)\triangle z(z).\label{eq:Ker-Ker}
\end{eqnarray}

From \eqref{eq:Ker-Ker}, it follows that,
\[
\triangle e=TS^{\dag}\triangle z=R\triangle z,
\]
or equivalently,
\begin{equation}
\triangle e_{t}=\sum_{s\in\mathbb{Z}}R_{s}\triangle z_{t-s}.\nonumber
\end{equation}

Thus, combining with the boundness of $\left\Vert R\right\Vert_{1,\mathrm{sp}}$ in Lemma~\ref{lem:boundlem}, for all $k\in\mathbb{Z}$,
\begin{eqnarray*}
\left\Vert \triangle e_{k}\right\Vert_{2}&\leq&\sum_{s\in\mathbb{Z}}\left\Vert R_{s}\right\Vert _{\mathrm{sp}}\left\Vert \triangle z_{k-s}\right\Vert _{2}\leq\left\Vert R\right\Vert_{1,\mathrm{sp}}\left\Vert \triangle z\right\Vert_{\infty,2} \\
&\leq& \left\Vert R\right\Vert_{1,\mathrm{sp}}\delta,
\end{eqnarray*}
and it completes the proof.
\end{proof}
\begin{rem}
From \eqref{eq:zT} and \eqref{eq:zz}, the quantity $T(z)S^{\dag}(z)$ can be roughly viewed as a transition function from $\triangle z(z)$ to $\triangle e(z)$. 
\end{rem}
\begin{rem}
Since
\begin{eqnarray*}
\triangle \hat{x}_{t+1} = (A+BL)\triangle \hat{x}_t + K\triangle z_{t+1},
\end{eqnarray*}
combining with that $A+BL$ is stable and $\Delta z_k$ is bounded due to the stealthy requirement in \eqref{eq:detect2}, we can easily get the bound of $\|\triangle \hat{x}_{t}\|_2$. Thus, the bound of $\|\triangle x_t\|_2$ follows from
\begin{eqnarray*}
\|\triangle x_t\|_2 &\leq& \|\triangle \hat{x}_t\|_2 + \|\triangle e_t\|_2.
\end{eqnarray*}
\end{rem}

\section{Simulation}\label{sec:simulate}

In this section, numerical examples are given to verify our proposed results in Theorem~\ref{thm:stable1}-\ref{thm:universal}. We first consider a double integrator from \cite{mo2016performance} below:
\begin{eqnarray}
x_{t+1} &=& \begin{bmatrix}1&0\\ 1&1\end{bmatrix}x_t + \begin{bmatrix}1\\ 0\end{bmatrix}u_t + B^au_t^a + w_t,\nonumber\\
y_t &=& x_t + \Gamma^a y_t^a + v_t.\label{eq:exam1}
\end{eqnarray}
The stationary estimator gain is given by $K=\begin{bmatrix}0.6&0\\ -1.4&1.6\end{bmatrix}$. The attack is stealthy if
\begin{eqnarray*}
\|\triangle z_t\|\leq 1
\end{eqnarray*}
for all $t\in \mathbb{N}$, and strictly stealthy if
\begin{eqnarray*}
\|\triangle z_t\|=0.
\end{eqnarray*}

In the rest of section, we will show that the different choices of sensor attack matrix $\Gamma^a$ and actuator attack matrix $B^a$ will make the system vulnerable, strictly vulnerable or invulnerable.

1) The attack matrices $B^a=\begin{bmatrix}1\\ 0\end{bmatrix},\Gamma^a=\begin{bmatrix}1&0\\0&1\end{bmatrix}$:

Following Theorem~\ref{thm:stable1}, the system is strictly vulnerable. Thus, a strictly stealthy attack sequence is designed by \eqref{eq:striv} and its effects on the norms of $\triangle e_t$ and $\triangle z_t$ are shown in Figure~\ref{strict_vulnerable}.
\begin{figure}[ht]
\begin{centering}
\includegraphics[width=8cm]{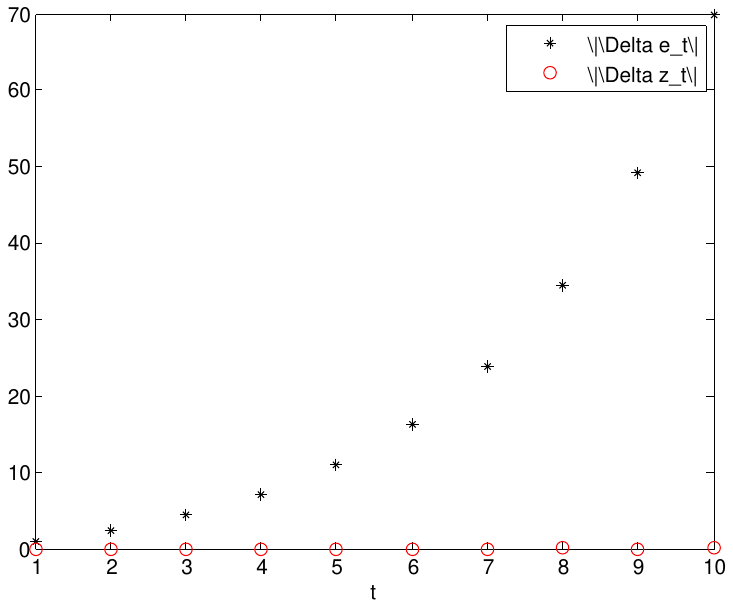}
\par\end{centering}
\caption{The evolution of $\triangle e_t$ and $\triangle z_t$}
\label{strict_vulnerable}
\end{figure}

From Figure~\ref{strict_vulnerable}, the system is destabilized by an attack while the residue bias is always zero, which confirms the strict vulnerability criterion in Theorem~\ref{thm:stable1}.

2) The attack matrices $B^a=0,\Gamma^a=\begin{bmatrix}0\\1\end{bmatrix}$:

Following Theorems~\ref{thm:stable1} and~\ref{thm:stable}, the system is vulnerable but not strictly vulnerable. A stealthy attack sequence is designed by \eqref{eq:unstableattack2}, and the norms of $\triangle e_t$ and $\triangle z_t$ are plotted in Figure~\ref{attack_diverge}.
\begin{figure}[ht]
\begin{centering}
\includegraphics[width=8cm]{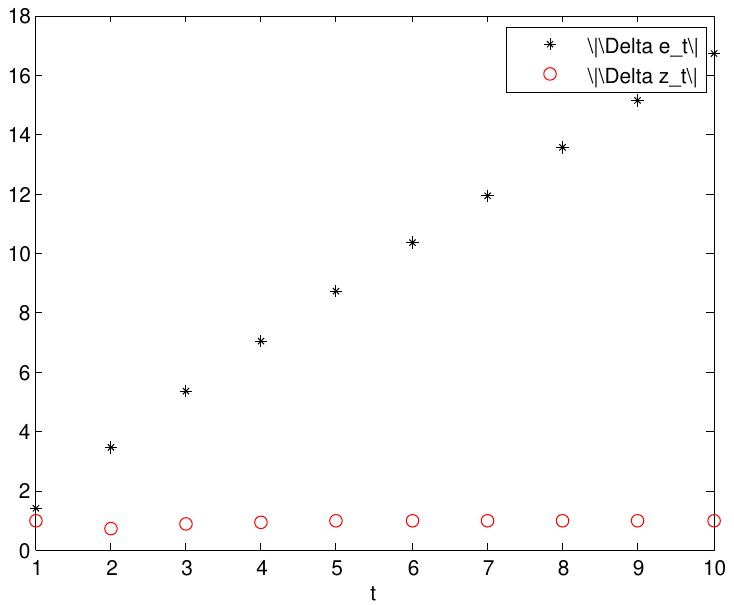}
\par\end{centering}
\caption{The evolution of $\triangle e_t$ and $\triangle z_t$.}
\label{attack_diverge}
\end{figure}

From Figure~\ref{attack_diverge}, the estimation error bias between the healthy and attacked system diverges while the residual bias is kept bounded. This confirms the vulnerability criterion in Theorem~\ref{thm:stable1}.

3) The attack matrices $B^a=0,\Gamma^a=\begin{bmatrix}1\\0\end{bmatrix}$:

Following Theorem~\ref{thm:stable}, the system is invulnerable for all stealthy attacks. We plot the reachable set for $\triangle e_t$ under all possible stealthy attacks in Figure~\ref{reach_invulner2_uni2} to show the system's robustness. Moreover, in the same figure, we also curve the universal bound for $\triangle e_t$ from Theorem~\ref{thm:universal}.
\begin{figure}[ht]
\begin{centering}
\includegraphics[width=8cm]{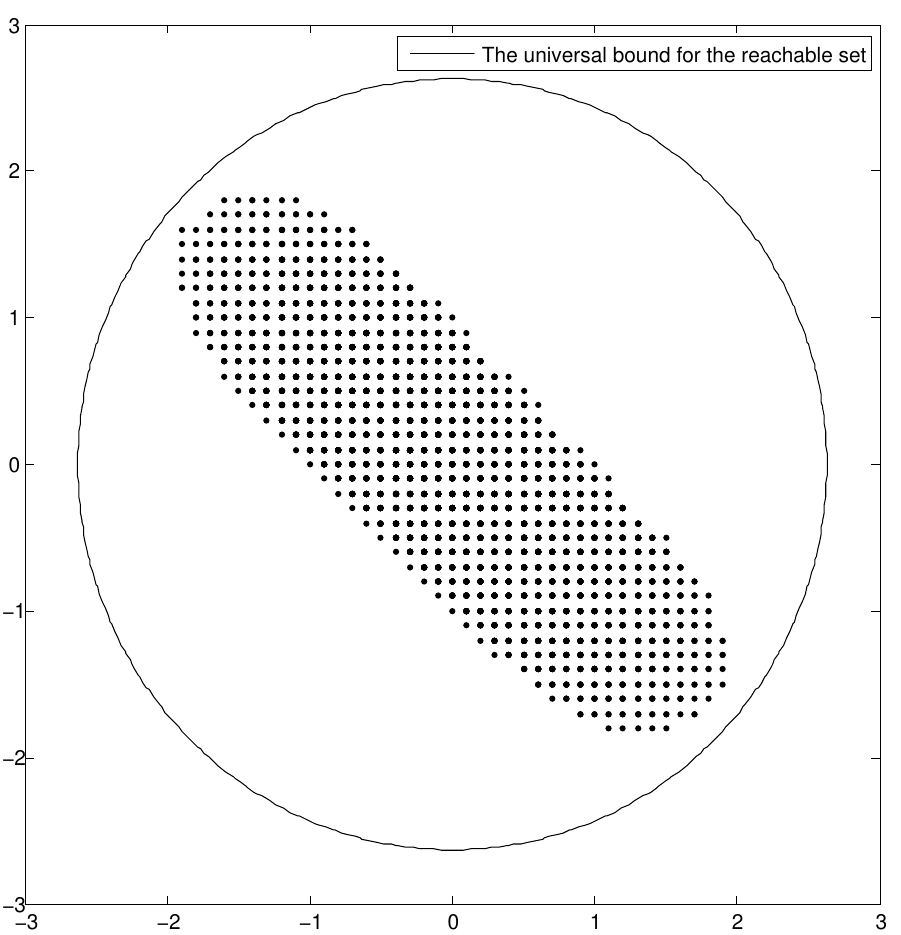}
\par\end{centering}
\caption{The reachable set of $\triangle e_t$ and its universal bound.}
\label{reach_invulner2_uni2}
\end{figure}

The dots in Figure~\ref{reach_invulner2_uni2} make up the reachable set of $\triangle e_t$. From Figure~\ref{reach_invulner2_uni2}, the estimation error bias for invulnerable system is always bounded under all stealthy attacks and its universal bound from Theorem~\ref{thm:universal} is tight and effective.

The system in \eqref{eq:exam1} is simply designed to illustrate the strict/non-strict vulnerability and invulnerability properties. To show our analysis for practical system, we have introduced the well-known Tennessee Eastman Process for further simulation. Tennessee Eastman Process (TEP) is a commonly used process proposed by Downs and Vogel in~\cite{downs1993plant}. In this simulation, we adopt a simplified version of TEP from~\cite{ricker1993model}, as follows:
\begin{align}
\dot{x} &= Ax+ Bu + B^au^a + w,\nonumber\\
y &= Cx + \Gamma^ay^a + v,\label{eq:exam2}
\end{align}
where $A,B$ and $C$ are constant matrices~\footnote{For more details about this dynamic model, please refer to Appendix~I in~\cite{ricker1993model}.}.

The TEP system is a MIMO system of order $n = 8$ with $p = 4$ inputs and $m = 10$ outputs. We discretize the system using the control system toolbox in MATLAB by selecting a sample period of $1$ second. Moreover, we take $B^a=\begin{bmatrix}0&0&1&0&0&0&0&0\end{bmatrix}^{\top}$ and $\Gamma^a=\begin{bmatrix}0&0&1&0&0&0&0&0&1&0\end{bmatrix}^{\top}$. Moreover, the covariance matrices $Q$ for $w$ and $R$ for $v$ are assumed to be identity matrices with proper dimensions.

Similar to that in Figure~\ref{reach_invulner2_uni2}, we compute the 8-dimensional reachable set of the TEP model and project it onto a 2-dimensional plant and show that our universal bound from Theorem~\ref{thm:universal} is still effective.
\begin{figure}[ht]
\begin{centering}
\includegraphics[width=8cm]{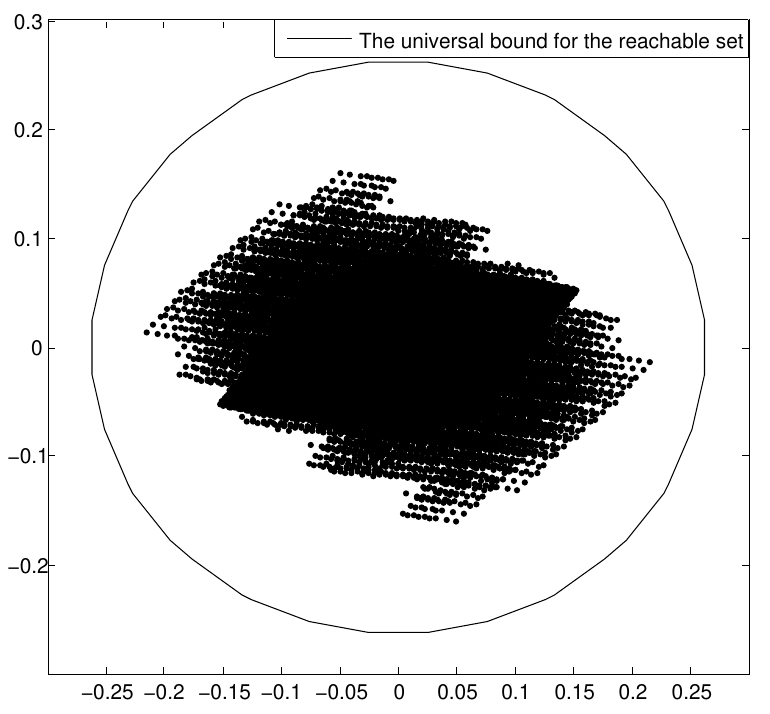}
\par\end{centering}
\caption{The reachable set of $\triangle e_t$ and its universal bound for the TEP Model.}
\label{TEP_reach}
\end{figure}


\section{Conclusion}\label{sec:conclude}
In this paper, the definitions of vulnerable and strictly vulnerable systems have been given for a stochastic linear system. A system is strictly vulnerable means that it can be destabilized by an attack that have no influence on the residue. Meanwhile, a system is vulnerable means that it can be destabilized by an attack that have bounded influence on the residue. The necessary and sufficient vulnerability and strict vulnerability conditions have been provided in this paper, which ensure the stability under stealthy and strictly stealthy attacks, respectively. Furthermore, for an invulnerable system, a performance bound for the bias between healthy and attacked system has also been given. The vulnerability condition shows what kind of system is robust to stealthy attacks and the performance bound shows how performance is affected by the stealthy attacks.

\appendices
\section{Proof of Lemma~\ref{lem:eq}}\label{app:eq}
{\bf Necessity:} The system in \eqref{eq:standard} being not invertible means that there exist a nonzero input $\{u_k: k\in \mathbb{N}\}$ such that $x_0=0,y_k=0$ for all $k\in \mathbb{N}$. Then, recall that in \eqref{eq:standardK}, we have
\begin{eqnarray}
x_{k+1}'&=&(A+KC)x_k'+(B+KD)u_k'\nonumber\\
&=&Ax_k'+Bu_k'+Ky_k'.\label{eq:attxy}
\end{eqnarray}

This implies that taking $u_k'=u_k$ for all $k\in \mathbb{N}$ will make $x_0'=0,y_k'=0$ for all $k\in \mathbb{N}$. Thus, the system in \eqref{eq:standardK} is also not invertible.

{\bf Sufficiency:} Suppose that there exist a nonzero input $\{u_k': k\in \mathbb{N}\}$ such that $x_0'=0,y_k'=0$. Recalling \eqref{eq:attxy}, we have
\begin{eqnarray*}
y_k' &=& Cx_k'+Du_k'=0,\\
x_{k+1}' &=& (A+KC)x_k'+(B+KD)u_k'\\
&=& Ax_k'+Bu_k'+Ky_k'=Ax_k'+Bu_k'.
\end{eqnarray*}

This will make $x_k=x_k',y_k=0$ in \eqref{eq:standardK} for all $k$ by letting $u_k=u_k'$ for all $k$, i.e., the system in \eqref{eq:standard} is not invertible.

\section{Proof of Lemma~\ref{lem:invariant}}\label{app:invariant}
We will first prove the uniqueness of $u$ by contradiction. Suppose there exists $u\neq u' $, such that
\begin{eqnarray*}
&&Ax + Bu\in V^*,\,Cx+Du = 0,\\
&&Ax + Bu'\in V^*,\,Cx+Du' = 0.
\end{eqnarray*}
By linearity of the system, we have
\begin{align*}
  0 + B (u-u') = x_1 \in V^*, D(u-u') = 0.
\end{align*}

By the property of the invariant set, there exist $u_k$ ($k\geq 1$) and corresponding $x_k$ ($k\geq 2$), such that
\begin{align*}
  Ax_k + Bu_k = x_{k+1} \in V^*,\,Cx_k+Du_k = 0,\forall k \geq 1.
\end{align*}

As a result, the non-zero control input sequence $u_0 = u-u',\, u_1,\,u_2,\ldots$ results in zero output for the system, which contradicts with the assumption that the system is invertible.
Thus, $u$ satisfying \eqref{xxx} is unique for each $x\in V^*$.

Now we show the existence of $Q$. Suppose that the basis of $V^*$ is given by $\{x_1^*,x_2^*,\ldots,x_{n^*}^*\}$, then there exists a unique set $\{u_1^*,u_2^*,\ldots,u_{n^*}^*\}$ such that
\begin{align*}
  Ax_i^* + Bu_i^* &\in V^*, \\
  Cx_i^* + Du_i^* &= 0
\end{align*}
for any $i=1,2,\ldots,n^*$.

The matrix $Q$ is defined as a transition matrix from $\{x_1^*,x_2^*,\ldots,x_{n^*}^*\}$ to $\{u_1^*,u_2^*,\ldots,u_{n^*}^*\}$, i.e., $u_i^*=Qx_i^*$ for any $i=1,2,\ldots,n^*$.
Taking arbitrary $x=a_1x_1^*+\ldots+a_{n^*}x_{n^*}^*$, we have $u=a_1u_1^*+\ldots+a_{n^*}u_{n^*}^*$ such that
\begin{eqnarray*}
&&Ax + Bu\in V^*,\,Cx+Du = 0,
\end{eqnarray*}
and $u=Qx$ always hold.

\section{Proof of Lemma~\ref{lem:suf}}\label{app:suf}
Recall the notations in Definition~\ref{defn:zero}, since the unstable eigenvector $v$ is reachable for $(A,B)$, there exists $u_0,u_1,\ldots,u_{n-1}$ such that
$$x_n = v.$$

Moreover, we can manipulate the magnitude of $v$ to satisfy that $\|y_k\|\leq \delta$ for all $k=0,1,\ldots,n-1$.

Then, we separate the proof into two cases:

1) Suppose  $|\lambda|>1$. A sequence of input is designed as
\begin{eqnarray}
{u}_{k} = \lambda^{k-n}Qv\label{eq:unstableattack}
\end{eqnarray}
for any $k\geq n$.

Under the input designed above, it follows directly that ${x}_{k}=\lambda^{k-n}v$ and $y_k = 0$ for any $k\geq n$.

Hence, we have  $\limsup_{k\rightarrow \infty}\|{x}_k\|=\infty$ and $\|y_k\|\leq \delta$ for all $k\in \mathbb{N}$.

2) Suppose  $|\lambda|=1$. A sequence of input is designed as
\begin{eqnarray}
\acute{u}_{kn+j} = k\lambda^{nk+j-n}Qv+\lambda^{nk}u_j\label{eq:unstableattack2}
\end{eqnarray}
for any $k\in \mathbb{N}$ and $j=0,1,\ldots,n-1$.

We derive the expression of corresponding $\acute{x}_{kn+j}$ by induction.
Suppose  $\acute{x}_{kn+j} = k\lambda^{nk+j-n}v+\lambda^{nk}x_j$. Then,
\begin{eqnarray*}
\acute{x}_{kn+j+1} &=& A\acute{x}_{kn+j} + B\acute{u}_{kn+j}\\
&=& A[k\lambda^{nk+j-n}v+\lambda^{nk}x_j] \\
&&+ B[k\lambda^{nk+j-n}Qv+\lambda^{nk}u_j]\\
&=& k\lambda^{nk+j-n}(A+BQ)v + \lambda^{nk}(Ax_j+Bu_j)\\
&=& k\lambda^{nk+j+1-n}v+\lambda^{nk}x_{j+1}
\end{eqnarray*}
and
\begin{eqnarray*}
\acute{x}_{(k+1)n} &=& A\acute{x}_{kn+n-1} + B\acute{u}_{kn+n-1}\\
&=& A[k\lambda^{nk-1}v+\lambda^{nk}x_{n-1}] \\
&&+ B[k\lambda^{nk-1}Qv+\lambda^{nk}u_{n-1}]\\
&=& k\lambda^{nk-1}(A+BQ)v + \lambda^{nk}(Ax_{n-1}+Bu_{n-1})\\
&=& k\lambda^{nk}v+\lambda^{nk}v\\
&=& (k+1)\lambda^{(k+1)n+0-n}v + \lambda^{(k+1)n}x_0.
\end{eqnarray*}

Thus, $\acute{x}_{kn+j} = k\lambda^{nk+j-n}v+\lambda^{nk}x_j$ for any $k\in \mathbb{N}$ and $j=0,1,\ldots,n-1$ is proved. This further implies that
\begin{eqnarray*}
\acute{y}_{kn+j} &=& C\acute{x}_{kn+j} + D\acute{u}_{kn+j}\\
&=& C[k\lambda^{nk+j-n}v+\lambda^{nk}x_j]\\
&&+D[k\lambda^{nk+j-n}Qv+\lambda^{nk}u_j]\\
&=& k\lambda^{nk+j-n}[C+DQ]v + \lambda^{nk}[Cx_j+Du_j]\\
&=& \lambda^{nk}y_{j}.
\end{eqnarray*}

Since $|\lambda|=1$, we have that $\|\acute{y}_{kn+j}\|=\|y_j\|\leq \delta$ for any $k\in \mathbb{N}$ and $j=0,1,\ldots,n-1$. Then we conclude that $\limsup_{k\rightarrow \infty}\|\acute{x}_k\|=\infty$ and $\|\acute{y}_k\|\leq \delta$ for all $k\in \mathbb{N}$.

It is worth noting that the control input we have designed can be complex valued, as the eigenvalue and eigenvector of $A+BQ$ may be complex valued. However, by linearity, we know that if we inject the real (imaginary) part of the designed sequence $u_k$ instead, then the state will be corresponding to the real (imaginary) part of $x_k$. Therefore, the divergence under a complex input means that either the real or the imaginary part of input can cause the divergence of state. Thus, we can choose the real or the imaginary part of that input as a real value input to make the system unstable.

\section{Proof of Lemma~\ref{lem:noninto}}\label{app:nes1}
Since the system in \eqref{eq:standard} is non-invertible, there exists $T\in \mathbb{N}$ and a nonzero stealthy input sequence $[u_0=u_0^*,\ldots,u_T=u_T^*]$ with $u_0^*\neq 0$ such that $[x_1=x_1^*,\ldots,x_T=x_T^*,x_{T+1}=a_1x_1^*+\ldots+a_Tx_T^*]$ and $y_k=0$ for all $k\in \mathbb{N}$.

Then, with $u_0=0$, we have that
\begin{eqnarray*}
&&u_1=0,\ldots,u_{T-1}=0,u_T=-a_1u_0^*\\
&&\Rightarrow x_{T+1}=-a_1x_1^*\\
&&u_1=0,\ldots,u_{T-1}=-a_2u_0^*,u_T=-a_2u_1^*\\
&&\Rightarrow =x_{T+1}=-a_2x_2^*\\
&&\vdots\\
&&u_1=-a_Tu_0^*,\ldots,u_{T-1}=-a_Tu_{T-2}^*,u_T=-a_Tu_{T-1}^*\\
&&\Rightarrow x_{T+1}=-a_Tx_T^*.
\end{eqnarray*}

Let
$$\bar{u}_0=u_0^*,\bar{u}_1=u_1^*-a_Tu_0^*,\ldots,\bar{u}_T=u_T^*-\sum_{i=0}^{T-1}a_{i+1}u_i^*$$
and its corresponding state sequence is denoted by
$$\bar{x}_1,\bar{x}_2,\ldots,\bar{x}_{T+1}.$$

Based on the combination property of linear system, the input $[u_0=\bar{u}_0,u_1=\bar{u}_1,\ldots,u_T=\bar{u}_T]$ is nonzero and stealthy, which makes $x_{T+1}=\bar{x}_{T+1}=0$.

For any $\lambda\geq 1$, we take
$$x=\sum_{k=0}^T\lambda^{-k}\bar{x}_k,~u=\sum_{k=0}^T\lambda^{-k}\bar{u}_k$$
and there exists a matrix $Q$ such that
$$u=Qx.$$

Then,
\begin{eqnarray*}
(A+BQ)x&=&Ax+Bu=\sum_{k=0}^T\lambda^{-k}(A\bar{x}_k+B\bar{u}_k)\\
&=&\sum_{k=0}^T\lambda^{-k}\bar{x}_{k+1}=\sum_{k=1}^{T+1}\lambda^{-k+1}\bar{x}_{k}\\
&=&\lambda\sum_{k=0}^T\lambda^{-k}\bar{x}_{k}=\lambda x
\end{eqnarray*}
and
\begin{eqnarray*}
(C+DQ)x=Cx+Du&=&\sum_{k=0}^T\lambda^{-k}(C\bar{x}_k+D\bar{u}_k)=0.
\end{eqnarray*}

Combining with that $x=\sum_{k=0}^T\lambda^{-k}\bar{x}_k$ is reachable for $(A,B)$, the conditions for the {\em unstable reachable zero-dynamic} of \eqref{eq:standardQ} in Definition~\ref{defn:zero} are satisfied and it completes the proof.

\section{Proof of Lemma~\ref{lem:nes}}\label{app:nes2}
Firstly, we will show the boundness of $\frac{\|u_k\|}{\|x_k\|+1}$.

Since the system in \eqref{eq:standard} is invertible, based on the results in Corollary~1 and Lemma~1 of \cite{Sain1969}, we have
\begin{eqnarray*}
u_k = \sum_{i=0}^{n-1}P_i[y_{k+i}-CA^ix_k],
\end{eqnarray*}
where $P_i,i=0,1,\ldots,n-1$ are the gains to reconstruct the input through outputs.
Thus, it follows that
\begin{eqnarray*}
\frac{\|u_k\|}{\|x_k\|+1} &\leq& \sum_{i=0}^{n-1}\|P_i\|[\frac{\|y_{k+i}\|}{\|x_k\|+1}+\|CA^i\|\frac{\|x_k\|}{\|x_k\|+1}]\\
&\leq& \sum_{i=0}^{n-1}\|P_i\|(\delta+\|CA^i\|).
\end{eqnarray*}

Then, there exists $U>0$ such that $\frac{\|u_k\|}{\|x_k\|+1}\leq U$ for all $k\in \mathbb{N}$.

Denote the state $x_k$ under stealthy input sequence $\{u_{i,t}:t\in \mathbb{N}\}$ by $x_{i,k}$. Taking arbitrary $N\in \mathbb{N}$ and $P>0$, for the vulnerability of system \eqref{eq:standard}, we could choose a cluster $\{\{u_{i,t}:t\in \mathbb{N}\}:i=1,2,\ldots\}$ such that\footnote{The peak sequence in \eqref{eq:peak} can be designed by
\begin{eqnarray*}
i_0=0,i_{k+1}=\min\{j:\|x_j\|>\|x_{i_k}\|\}
\end{eqnarray*}
if there exist stealthy $\{u_t:t\in \mathbb{N}\}$ such that $\limsup_{k\rightarrow \infty}\|{x}_{k}\|=\infty$}
\begin{eqnarray}
\Vert x_{1,k_1}\Vert &>& P,\nonumber\\
\Vert x_{i+1,k_{i+1}+q}\Vert &>& (i+1)P,~\text{for all}~q=-N,\ldots,n,\nonumber\\
\Vert x_{i+1,k_{i+1}}\Vert &>& \max_{q=-N,\ldots,n}(\Vert x_{i,k_i}\Vert,\Vert x_{i+1,k_{i+1}+q}\Vert).\label{eq:peak}
\end{eqnarray}

And it follows directly that $\limsup_{i\rightarrow \infty}\|{x}_{i,k_i+q}\|=\infty$ for any $q=-N,\ldots,n$.
%
%

Based on the Bolzano-Weierstrass Theorem~\cite{Bartle2000real}, there exists a convergent subsequence for a bounded sequence. Since
\begin{eqnarray*}
\|\frac{u_{i,k}}{\|x_{i,k}\|+1}\|\leq U,\|\frac{x_{i,k}}{\|x_{i,k}\|+1}\|\leq 1,\forall k\in \mathbb{N},
\end{eqnarray*}
there exists a subsequence of $\{i:i\in \mathbb{N}\}$, i.e., $\{j_i:i\in \mathbb{N}\}\subseteq \{i:i\in \mathbb{N}\}$, such that
\begin{eqnarray*}
\lim_{i\rightarrow \infty}\frac{x_{j_i,k_{j_i}+q}}{\|x_{j_i,k_{j_i}+q}\|+1}=\check{x}_q,q=-N,\ldots,0,\ldots,n\\
\lim_{i\rightarrow \infty}\frac{u_{j_i,k_{j_i}+q}}{\|x_{j_i,k_{j_i}+q}\|+1}=\check{u}_q,q=-N,\ldots,0,\ldots,n.
\end{eqnarray*}

For any $q=-N,\ldots,n-1$, we have
\begin{eqnarray}
A\check{x}_q+B\check{u}_q&=&\lim_{i\rightarrow \infty}\frac{Ax_{j_i,k_{j_i}+q}+Bu_{j_i,k_{j_i}+q}}{\|x_{j_i,k_{j_i}+q}\|+1}\nonumber\\
&=& \lim_{i\rightarrow \infty}\frac{x_{j_i,k_{j_i}+q+1}}{\|x_{j_i,k_{j_i}+q}\|+1}\nonumber\\
&=& \underbrace{\lim_{i\rightarrow \infty}\frac{\|x_{j_i,k_{j_i}+q+1}\|+1}{\|x_{j_i,k_{j_i}+q}\|+1}}_{c_q}\check{x}_{q+1},\label{eq:checkx}
\end{eqnarray}
and
\begin{eqnarray}
C\check{x}_q+D\check{u}_q&=&\lim_{i\rightarrow \infty}\frac{Cx_{j_i,k_{j_i}+q}+Du_{j_i,k_{j_i}+q}}{\|x_{j_i,k_{j_i}+q}\|+1}\nonumber\\
&=& \lim_{i\rightarrow \infty}\frac{y_{j_i,k_{j_i}+q}}{\|x_{j_i,k_{j_i}+q}\|+1}\nonumber\\
&=& 0\label{eq:checky}
\end{eqnarray}
as $\|y_{j_i,k_{j_i}+q}\|\leq \delta$ and $\limsup_{i\rightarrow \infty}\|x_{j_i,k_{j_i}+q}\|=\infty$.

Since the state $x_k$ is $n$-dimensional, the vectors $\check{x}_{-N},\ldots,\check{x}_{0},\ldots,\check{x}_{n}$ are linearly dependent. Then, a subspace $V$ is designed by
\begin{eqnarray*}
V = \text{span}[\check{x}_{-N},\ldots,\check{x}_{0},\ldots,\check{x}_{d}],
\end{eqnarray*}
where $0\leq d\leq n-1$ such that $\text{span}[\check{x}_{-N},\ldots,\check{x}_{0},\ldots,\check{x}_{d}]=\text{span}[\check{x}_{-N},\ldots,\check{x}_{0},\ldots,\check{x}_{d+1}]$, i.e., $\check{x}_{d+1}\in V$.

Then, for any $\check{x}\in V$ and let $\check{x}= b_{-N}\check{x}_{-N}+\ldots+b_{d}\check{x}_{d}$, we have
\begin{eqnarray*}
&&A\check{x} + B[b_{-N}\check{u}_{-N}+\ldots+b_{d}\check{u}_{d}]\\
&=& b_{-N}[A\check{x}_{-N}+B\check{u}_{-N}]+\ldots+b_{d}[A\check{x}_{d}+B\check{u}_{d}]\\
&=& b_{-N}c_{-N}\check{x}_{-N+1}+\ldots+b_{d}c_d\check{x}_{d+1}\in V
\end{eqnarray*}
and
\begin{eqnarray*}
&&C\check{x}+D[b_{-N}\check{u}_{-N}+\ldots+b_{d}\check{u}_{d}]\\
&=& b_{-N}[C\check{x}_{-N}+D\check{u}_{-N}]+\ldots+b_{d}[C\check{x}_{d}+D\check{u}_{d}]\\
&=& 0.
\end{eqnarray*}

Thus the subspace $V$ is an invariant set satisfying \eqref{eq:invariant1}. Combining with that $\check{x}_{-N},\ldots,\check{x}_{d}$ are reachable for $(A,B)$, we have $V\subseteq V^*$.

Based on the Lemma~\ref{lem:invariant}, since $\check{x}_{-N},\ldots,\check{x}_{0}\in V^*$, it follows that there exists matrix $Q$ such that
\begin{eqnarray}\label{eq:checkN}
\check{u}_{q}=Q\check{x}_{q}, \forall q=-N,\ldots,0.
\end{eqnarray}

At last, we will show that $A+BQ$ is unstable on $V^*$ by contradiction. Suppose that $A+BQ$ is stable on $V^*$, for that $N$ can be arbitrary large, there exists an integer $p\leq N$ such that
\begin{eqnarray}
\|(A+BQ)^p\upsilon\|<\upsilon\label{eq:unstableass}
\end{eqnarray}
for all $\upsilon\in V^*$.

From \eqref{eq:checkN}, it follows that
\begin{eqnarray*}
A\check{x}_{q}+B\check{u}_q &=& \lim_{i\rightarrow \infty}\frac{\|x_{j_i,k_{j_i}+q+1}\|+1}{\|x_{j_i,k_{j_i}+q}\|+1}\check{x}_{q+1}\\
&=& (A+BQ)\check{x}_{q}.
\end{eqnarray*}
for all $q=-p,-p+1,\ldots,-1$.

The above further implies that
\begin{eqnarray}
&&\lim_{i\rightarrow \infty}\frac{\|x_{j_i,k_{j_i}}\|+1}{\|x_{j_i,k_{j_i}-p}\|+1}\check{x}_{0} \nonumber\\
&=& \lim_{i\rightarrow \infty}\frac{\|x_{j_i,k_{j_i}}\|+1}{\|x_{j_i,k_{j_i}-1}\|+1}\cdots\frac{\|x_{j_i,k_{j_i}-p+1}\|+1}{\|x_{j_i,k_{j_i}-p}\|+1}\check{x}_{0}\nonumber\\
&=& (A+BQ)^p\check{x}_{-p}.\label{eq:itercheck}
\end{eqnarray}

Based on the peak property defined in \eqref{eq:peak}, we have that $\lim_{k\rightarrow \infty}\frac{\|x_{j_i,k_{j_i}}\|+1}{\|x_{j_i,k_{j_i}-p}\|+1}\geq 1$.

Hence, together with $\|\check{x}_{-p}\|=\|\check{x}_{0}\|=1$, the equation \eqref{eq:itercheck} induces that
\begin{eqnarray*}
\|\check{x}_{-p}\| &=& \|\check{x}_{0}\|\\
&\leq& \|\lim_{i\rightarrow \infty}\frac{\|x_{j_i,k_{j_i}}\|+1}{\|x_{j_i,k_{j_i}-p}\|+1}\check{x}_{0}\|\\
&=& \|(A+BQ)^p\check{x}_{-p}\|.
\end{eqnarray*}

The inequality above contradicts with the assumption in \eqref{eq:unstableass}, thus $A+BQ$ is unstable on $V^*$ and one of its unstable eigenvector $v\in V^*$.

Since $v\in V^*\subseteq \text{span}\begin{bmatrix}B&AB&\ldots&A^{n-1}B\end{bmatrix}$, based on the Lemma~\ref{lem:invariant}, the conditions for the {\em unstable reachable zero-dynamic} of \eqref{eq:standardQ} in Definition~\ref{defn:zero} are satisfied and the proof is finished.

\section{Proof of Lemma~\ref{lem:boundlem}}\label{app:bound}
%
Since the system in \eqref{eq:ss1}-\eqref{eq:ss2} is invulnerable, there exists $M>1$ such that\footnote{Suppose that there exists a sequence of attack such that $\left\Vert \Delta z\right\Vert_{\infty,2}=0$ and $0<\left\Vert \Delta e\right\Vert_{\infty,2}<\infty$, based on the linearity, there must exists another sequence of attack such that $\left\Vert \Delta z\right\Vert_{\infty,2}=0$ and $\left\Vert \Delta e\right\Vert_{\infty,2}=\infty$. It contradicts with invulnerability.}
\begin{equation}
\sup_{\left(y_t^a,u_t^a\right):0\leq \left\Vert \Delta z\right\Vert_{\infty,2}\leq \delta}f(\limsup_{t\rightarrow \infty}\Vert \Delta e_t\Vert_2,\limsup_{t\rightarrow \infty}\Vert \Delta z_t\Vert_2)\leq M,\label{eq:bound}
\end{equation}
where
\begin{eqnarray*}
f(a,b)=\begin{cases}
\frac{a}{b}, & \text{if}~b>0;\\
1, & \text{if}~a=0,b=0;\\
\infty, & \text{if}~a>0,b=0.
\end{cases}
\end{eqnarray*}

Firstly, we will prove
\begin{eqnarray}
\sup_{\left|z\right|=1}\sup_{0\leq \left\Vert S(z)\mu\right\Vert_{2}\leq \delta}f(\Vert T(z)\mu\Vert_2, \Vert S(z)\mu\Vert_2)\leq M\label{eq:cond2}
\end{eqnarray}
by contradiction. Suppose that there exists $\omega\in[-\pi,\pi)$ and $\mu$ such that
\begin{equation}
\frac{\left\Vert T\left(e^{j\omega}\right)\mu\right\Vert _{2}}{\left\Vert S\left(e^{j\omega}\right)\mu\right\Vert _{2}}>M.\label{eq:bigger}
\end{equation}

Let the attack input
\[
\zeta_{t}=e^{j\omega t}\mu~\text{for all}~t\in \mathbb{N}.
\]

It then follows from~\eqref{eq:zT} and \eqref{eq:zz} that
\begin{align*}
\limsup_{t\rightarrow \infty}\Vert \Delta e_{t}\Vert_2 & = \Vert T\left(e^{j\omega}\right)\mu\Vert_2,\\
\limsup_{t\rightarrow \infty}\Vert \Delta z_{t}\Vert_2 & = \Vert S\left(e^{j\omega}\right)\mu\Vert_2.
\end{align*}

Since the magnitude of $\mu$ will not change \eqref{eq:bigger}, thus we let $\|\mu\|_2$ small enough to make $\left\Vert S\left(e^{j\omega}\right)\mu\right\Vert_{2}\leq \left\Vert \Delta z\right\Vert _{\infty,2}\leq \delta$.

Hence, from~\eqref{eq:bigger},
\[
\frac{\limsup_{t\rightarrow \infty}\left\Vert \Delta e_{t}\right\Vert_{2}}{\limsup_{t\rightarrow \infty}\left\Vert \Delta z_{t}\right\Vert_{2}}=\frac{\left\Vert T\left(e^{j\omega}\right)\mu\right\Vert_{2}}{\left\Vert S\left(e^{j\omega}\right)\mu\right\Vert_{2}}>M,
\]
and it contradicts with \eqref{eq:bound}. Since the system is by assumption invulnerable, we must have \eqref{eq:cond2} and the result in \eqref{eq:cond} is thus proved.

Then, we will separate the proof for \eqref{eq:R1} into two steps:

Step 1) Let $S^{\ast}(z)$ be the conjugate transpose of $S(z)$ and
$A(z)=S(z)S^{\ast}(z)$. Let $\hat{m}\in\mathbb{N}$ be the maximum
number of linearly independent vectors $w_{i}(z)$, $i=1,\cdots,\hat{m}$,
satisfying
\[
A(z)w_{i}(z)=0\text{ for all }~|z|=1.
\]

Since the entries of $A(z)$ are rational functions, by solving the
above, it is easy to see that the entries of $w_{i}(z)$ can be rational
functions. Let $\tilde{m}=m-\hat{m}$ and $v_{i}(z)$, $i=1,\cdots,\tilde{m}$,
be a base for $\mathrm{span}^{\perp}\left\{ w_{i}(z),i=1,\cdots,\hat{m}\right\} $.
It is also easy to make that the entries of $v_{i}(z)$ are rational
functions. Let $V(z)=\left[v_{1},\cdots,v_{\tilde{m}}\right]$, $W(z)=\left[w_{1},\cdots,w_{\hat{m}}\right]$
and $U(z)=\left[V,W\right]$. We then have
\[
A(z)=U(z)\left[\begin{array}{cc}
\tilde{A}(z) & 0\\
0 & 0
\end{array}\right]U^{\ast}(z)
\]
with $\det\tilde{A}(z)\neq0$ for at least one $|z|=1$. Since
$\tilde{A}(z)=V^{\ast}(z)A(z)V(z)$ has rational entries, $\det\tilde{A}(z)$
is a rational function. Therefore, $\det\tilde{A}(z)\neq 0$ for almost
all $|z|=1$. It then follows that
\begin{align*}
S(z)S^{\dagger}(z) & =S(z)S^{\ast}(z)\left(S(z)S^{\ast}(z)\right)^{\dagger}\\
 & =A(z)A^{\dagger}(z)\\
 & =U(z)\left[\begin{array}{cc}
I_{\tilde{m}} & 0\\
0 & 0
\end{array}\right]U^{\ast}(z),
\end{align*}
where $S^{\dagger}(z)$ denotes the Moore-Penrose pseudoinverse $S(z)$
and $I_{\tilde{m}}$ denotes the identity matrix of dimension $\tilde{m}$.
Hence, the inverse $z$-transform $SS^{\dagger}$ of $S(z)S^{\dagger}(z)$ satisfies
\begin{equation}
\left\Vert SS^{\dagger}\right\Vert _{1,\mathrm{sp}}<\infty.\label{eq:SSpinv}
\end{equation}

Step 2) Our next step is to show that $\left\Vert R\right\Vert _{1,\mathrm{sp}}$ is finite. We do so by contradiction. Suppose that $\left\Vert R\right\Vert _{1,\mathrm{sp}}=\infty$.
For each $t\in\mathbb{Z}$, let $\eta_{t}$ be a vector satisfying
$\left\Vert \eta_{t}\right\Vert _{2}=1$ and
\begin{equation}
\left\Vert R_{t}\eta_{t}\right\Vert_{2}=\left\Vert R_{t}\right\Vert _{\mathrm{sp}}.\label{eq:achieve-norm}
\end{equation}

We have
\begin{eqnarray*}
\infty &=&\sum_{t\in\mathbb{Z}}\left\Vert R_{t}\right\Vert _{\mathrm{sp}}=\sum_{t\in\mathbb{Z}}\left\Vert R_{t}\eta_{t}\right\Vert _{2}\\
 &\leq&\sum_{t\in\mathbb{Z}}\left\Vert R_{t}\eta_{t}\right\Vert _{1}=\sum_{d=1}^{m}\sum_{t\in\mathbb{Z}}\left|\left(R_{t}\eta_{t}\right)_{d}\right|,
\end{eqnarray*}
where $\left(R_{t}\eta_{t}\right)_{d}$ is the $d$-th element of $R_{t}\eta_{t}$.

Hence, there exists $1\leq d\leq m$ such that
\[
\sum_{t\in\mathbb{Z}}\left|\left(R_{t}\eta_{t}\right)_{d}\right|=\infty.
\]

Then, there exists a sequence $\sigma_{t}\in\{-1,1\}$, $t\in\mathbb{Z}$,
satisfying
\begin{equation}
\sum_{t\in\mathbb{Z}}\sigma_{t}\left(R_{t}\eta_{t}\right)_{d}=\infty.\label{eq:blow}
\end{equation}

Given any $T$, let $\xi_{t}=\sigma_{T-t}\eta_{T-t}$ for all $t\in\mathbb{Z}$ and $S^{\dagger}:L_{2}^{m}\left(\mathbb{Z}\right)\rightarrow L_{2}^{m_a+p_a}\left(\mathbb{Z}\right)$
denote the Moore-Penrose pseudoinverse of $S$. Put
\begin{align*}
\zeta & =S^{\dagger}\xi,
\end{align*}
then
\begin{align*}
\Delta e & =T\zeta=R\xi,\\
\Delta z & =S\zeta=SS^{\dagger}\xi.
\end{align*}

Using~(\ref{eq:blow}) we get
\begin{eqnarray}
\lim_{T\rightarrow\infty}\left\Vert \Delta e_{T}\right\Vert_{2} &=&\lim_{T\rightarrow\infty}\left\Vert \sum_{t=-T}^{T}R_{t}\xi_{T-t}\right\Vert _{2}\nonumber\\
&=&\lim_{T\rightarrow\infty}\left\Vert \sum_{t=-T}^{T}\sigma_{t}R_{t}\eta_{t}\right\Vert _{2}=\infty.\label{eq:Dx}
\end{eqnarray}

Also, using~$(\ref{eq:SSpinv}),$ for all $t\in\mathbb{Z}$,
\begin{eqnarray}
\left\Vert \Delta z_{t}\right\Vert _{2} &=&\left\Vert \left(SS^{\dagger}\xi\right)_{t}\right\Vert _{2}\leq \sum_{s\in\mathbb{Z}}\left\Vert \left(SS^{\dagger}\right)_{s}\right\Vert _{\mathrm{sp}}\left\Vert \xi_{t-s}\right\Vert _{2}\nonumber \\
&\leq& \left\Vert SS^{\dagger}\right\Vert _{1,\mathrm{sp}}\left\Vert \xi\right\Vert _{\infty,2}<\infty.\label{eq:Dy}
\end{eqnarray}

From~(\ref{eq:Dx}) and~(\ref{eq:Dy}), the system is vulnerable.
Since by assumption the system is invulnerable, we must then have
$\left\Vert R\right\Vert _{1,\mathrm{sp}}<\infty$.

\bibliographystyle{unsrt}
\bibliography{refs}
\begin{IEEEbiography}[{\includegraphics[width=1in,height=1.25in,clip,keepaspectratio]{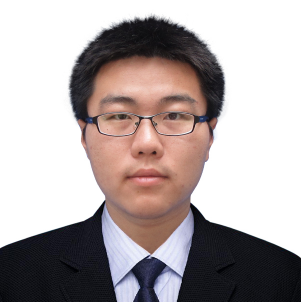}}]{Tianju Sui}
received B.S. and Ph.D. degrees from Zhejiang University, Hangzhou, China. in 2012 and 2017, respectively. He is currently serving as an Associate Professor in Dalian University of Technology. His main research area includes networked estimation, distributed estimation and the security of Cyber-Physical systems.
\end{IEEEbiography}
\begin{IEEEbiography}[{\includegraphics[width=1in,height=1.25in,clip,keepaspectratio]{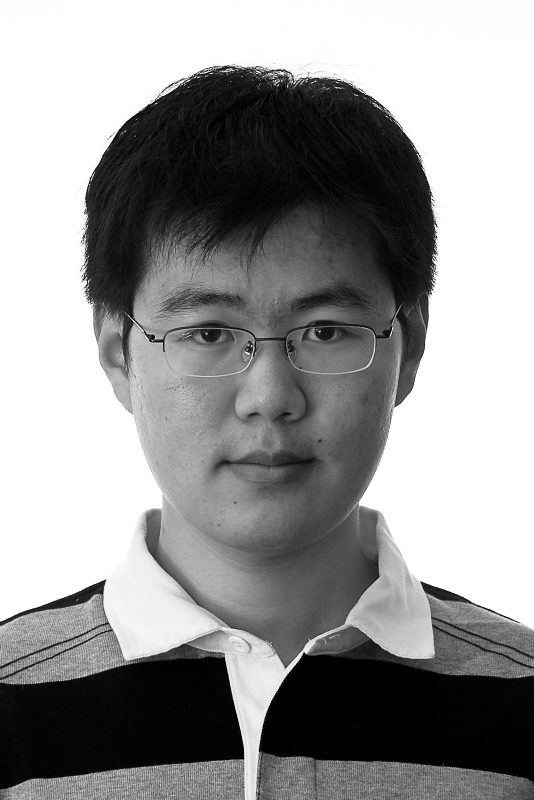}}]{Yilin Mo}
is an Associate Professor in the Department of Automation at Tsinghua University. He received his Ph.D. In Electrical and Computer Engineering from Carnegie Mellon University in 2012 and his Bachelor of Engineering degree from Department of Automation, Tsinghua University in 2007. Prior to his current position, he was a postdoctoral scholar at Carnegie Mellon University in 2013 and California Institute of Technology from 2013 to 2015. From 2015 to 2018, he held an assistant professor position in the School of Electrical and Electronic Engineering at Nanyang Technological University. His research interests include secure control systems and networked control systems, with applications in sensor networks and power grids.
\end{IEEEbiography}
\begin{IEEEbiography}[{\includegraphics[width=1in,height=1.25in,clip,keepaspectratio]{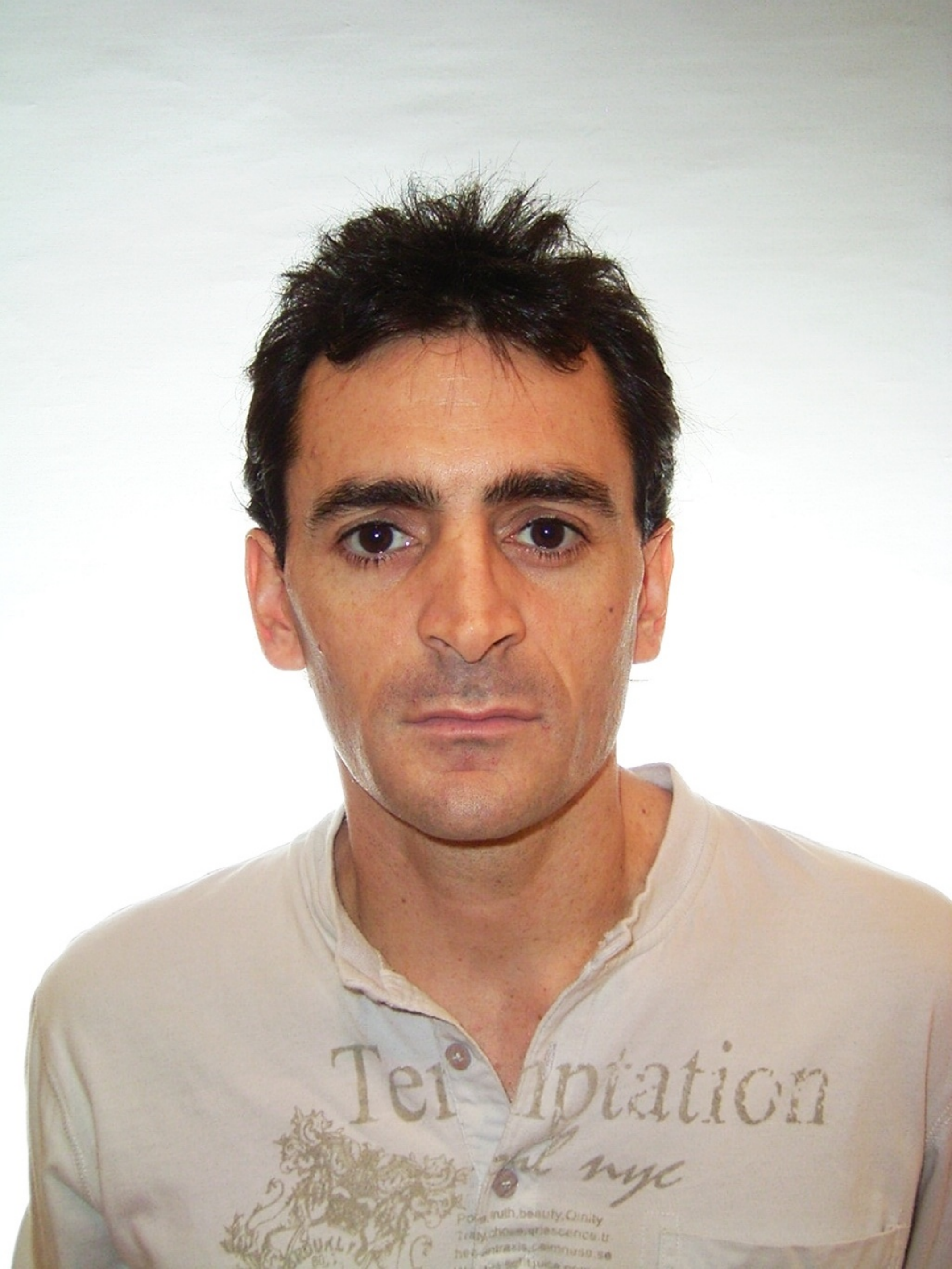}}]{Damian Marelli}
received his Bachelors Degree in Electronics Engineering from the Universidad Nacional de Rosario, Argentina in 1995. He also received his Bachelor (Honous) degree in Mathematics and Ph.D. degree in Electrical Engineering, both from the University of Newcastle, Australia in 2003. From 2004 to 2005 he was postdoc at the Laboratoire d'Analyse Topologie et Probabilites, CNRS/Universite de Provence, France. From 2005 to 2015 he was Research Academic at the Centre for Complex Dynamic Systems and Control, the University of Newcastle, Australia. In 2007 he received a Marie Curie Postdoctoral Fellowship, hosted at the University of Vienna, and in 2011 he received a Lise Meitner Senior Fellowship, hosted at the Austrian Academy of Sciences. Since 2016, he is Professor at the School of Automation, Guangdong University of Technology, China and Independent Researcher appointment at the French-Argentinean International Center for Information and Systems Sciences, National Scientific and Technical Research Council, Argentina.  His main research interests include system theory, statistical signal processing and distributed processing.
\end{IEEEbiography}
\begin{IEEEbiography}[{\includegraphics[width=1in,height=1.25in,clip,keepaspectratio]{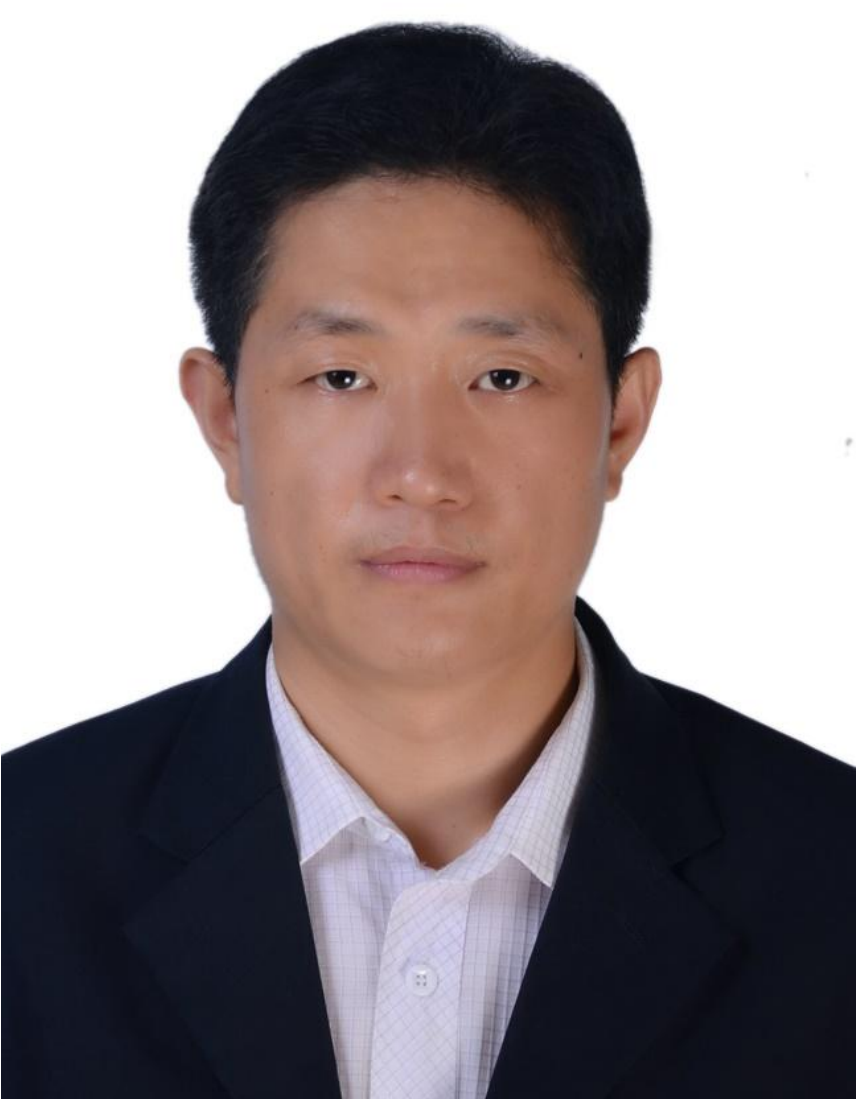}}]{Xi-Ming Sun}
received the Ph.D. degree in Control Theory and Control Engineering from the Northeastern University, China, in 2006. From August 2006 to December 2008, he worked as a Research Fellow in the Faculty of Advanced Technology, University of Glamorgan, UK. He then visited the School of Electrical and Electronic Engineering, Melbourne University, Australia in 2009, and Polytechnic Institute of New York University in 2011, respectively. He is IEEE Senior Member. He serves as   Associate Editor in IEEE Transactions on Cybernetics. He is currently a Professor in the School of Control Science and Engineering, Dalian University of Technology, China. He was awarded the Most Cited Article 2006-2010 from the journal of Automatica in 2011. His research interests include hybird systems, networked control systems, and nonlinear systems.
\end{IEEEbiography}
\begin{IEEEbiography}[{\includegraphics[width=1in,height=1.25in,clip,keepaspectratio]{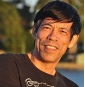}}]{Minyue Fu}
(F' 03) received the B.Sc. degree in electrical engineering from the University of Science and Technology of China, Hefei, China, in 1982, and the M.S. and Ph.D. degrees in electrical engineering from the University of Wisconsin-Madison, Madison, WI, USA. in 1983 and 1987, respectively.

From 1987 to 1989, he was an Assistant Professor in the Department of Electrical and Computer Engineering, Wayne State University, USA. He joined the Department of Electrical and Computer Engineering at the University of Newcastle, Australia, in 1989. Currently, he is a Chair Professor of Electrical Engineering. He has been a Visiting Associate Professor at the University of Iowa, USA, Nanyang Technological University, Singapore and Tokyo University, Tokyo, Japan. He has held a ChangJiang Visiting Professorship at Shandong University, Jinan, China, a Extinguished Professorship at Zhejiang University and Guangdong University of Technology, China. He has been an Associate Editor for the IEEE Transactions on Automatic Control, Automatica, IEEE Transactions on Signal Processing, and the Journal of Optimization and Engineering. He is a Fellow of IEEE, Fellow of Engineers of Australia and Fellow of Chinese Association of Automation. His current research areas include networked control systems, smart electricity networks, and super-precision positioning control systems.
\end{IEEEbiography}

\end{document}